\documentclass[ejs]{imsart}

\RequirePackage{amsthm,amsmath}
\RequirePackage[numbers]{natbib}
\RequirePackage[colorlinks,citecolor=blue,urlcolor=blue]{hyperref}

\RequirePackage{graphicx}

\usepackage{amssymb}
\usepackage{bm}
\usepackage{bbm}
\usepackage{listings}
\usepackage{caption}
\usepackage{tikz}
\usepackage{subcaption}
\usepackage{float}
\usepackage{natbib} 
\usepackage{url}
\usepackage{booktabs}  
\usepackage{multirow}
\usepackage{enumitem}
\theoremstyle{plain}
\newtheorem{definition}{Definition}
\newtheorem{theorem}{Theorem}[section]

\newtheorem{prop}{Proposition}
\newtheorem{remark}{Remark}
\newtheorem{lemma}[theorem]{Lemma}
\usepackage{algorithm}
\usepackage{algpseudocode}

\doi{https://doi.org/10.1214/24-EJS2234
}
\pubyear{2024}
\volume{18}
\firstpage{1455}
\lastpage{1494}

\startlocaldefs
\endlocaldefs

\begin{document}

\begin{frontmatter}

\title{Differentially Private Confidence Intervals for Proportions under Stratified Random Sampling\support{The research presented in this paper was supported by the U.S. Census Bureau Cooperative Agreement CB20ADR0160001.}}

\runtitle{Differentially Private Confidence Intervals}


\begin{aug}
\author{\fnms{Shurong} \snm{Lin}
\ead[label=e1]{shrlin@bu.edu}}
\address{Department of Mathematics and Statistics, Boston University, Boston, MA\\
\printead{e1}}

\author{\fnms{Mark} \snm{Bun}\ead[label=e2]{mbun@bu.edu}}
\and
\author{\fnms{Marco} \snm{Gaboardi}\ead[label=e3]{gaboardi@bu.edu}}

\address{Department of Computer Science, Boston University, Boston, MA\\
\printead{e2,e3}}

\author{\fnms{Eric D.} \snm{Kolaczyk}\ead[label=e4]{eric.kolaczyk@mcgill.ca}}
\address{Department of Mathematics and Statistics, McGill University, Canada\\
\printead{e4}}

\author{\fnms{Adam} \snm{Smith}\ead[label=e5]{ads22@bu.edu}}
\address{Department of Computer Science, Boston University, Boston, MA\\
\printead{e5}}

\end{aug}

\runauthor{S. Lin et al.}

\begin{abstract}
Confidence intervals are a fundamental tool for quantifying the uncertainty of parameters of interest. With the increase of data privacy awareness,  
developing
a private version of confidence intervals has gained growing attention from both statisticians and computer scientists. 
Differential privacy is a state-of-the-art framework for analyzing privacy loss when releasing statistics computed from sensitive data.
Recent work has been done around differentially private confidence intervals, yet to the best of our knowledge, rigorous methodologies on differentially private confidence intervals
in the context of survey sampling have not been studied.
In this paper, we propose three differentially private algorithms for constructing confidence intervals for proportions under stratified random sampling. We articulate two variants of differential privacy that make sense for data from stratified sampling designs, analyzing each of our algorithms within one of these two variants.
We establish analytical privacy guarantees and asymptotic properties of the estimators. 
In addition, we conduct simulation studies to evaluate the proposed private confidence intervals, and two applications to the 1940 Census data are provided.
\end{abstract}



\tableofcontents
\end{frontmatter}

\section{Introduction}
With the increase of privacy awareness in the modern information era, establishing
privacy-preserving methodologies for statistics and machine learning has become an active research area.
\textit{Differential privacy}, a state-of-the-art privacy protection technique \cite{dwork2006calibrating}, is considered a gold standard for rigorous privacy guarantees. 
Not only has it drawn significant attention in academia \cite{dwork2014algorithmic, DworkR16}, but also it has been deployed by governments, firms, and other data agencies, such as the U.S. Census Bureau \cite{uscensus},
Google \cite{google}, Microsoft \cite{DingKY17}, and Apple \cite{apple}.
Recently, the U.S. Census Bureau released a new demonstration of its differentially private Disclosure Avoidance System (DAS) for the 2020 Census \cite{censushandbook, Groshen2022Disclosure}. 
At the intersection of differential privacy and statistics, both statisticians and computer scientists are working on developing private versions of statistical inference procedures. 
Early work discussing differential privacy in the context of statistics includes \cite{dwork2010differential,dwork2009differential,tm08651,Smith09}. More recent work has explored statistical inference and estimation under the constraint of differential privacy
\cite{cai2021cost, kroll2021ondens, lam2022mini}.

As one of the most fundamental tools for statistical 
inference, confidence intervals
are ubiquitous in quantifying the uncertainty of parameters of interest.
In this paper, we propose three differentially private algorithms for constructing confidence intervals for the population proportion under stratified random sampling.
To the best of our knowledge, our work is the first to establish rigorous methodologies on differentially private confidence intervals in the context of survey sampling.
Survey sampling is an important area in statistics that involves selecting a sample of individuals from a target population to conduct a survey. It provides timely and cost-efficient estimates of population characteristics of interest and is widely used in broad-scale data gatherings, such as the American Community Survey (ACS), the Survey of Income and Program Participation (SIPP), and the Current Population Survey (CPS).

This paper provides the first study of differentially private confidence intervals for data from stratified sampling designs. Specifically:
\begin{itemize}
    \item 
    We articulate two specific variants of differential privacy that are appropriate for data from stratified sampling designs. In addition to the standard notion of differential privacy, we also consider settings in which the sample stratum sample sizes are fixed and public. This latter setting allows for simpler algorithms and tighter confidence intervals. 
    
    \item  We give methods to propagate the uncertainty due to the application of differentially private mechanisms (adding random noise)
    into the construction of confidence intervals. 
    A necessary bias correction is made to achieve (asymptotic) unbiased variance estimates. Central limit theorem (CLT)-type statements are provided to guarantee the confidence level asymptotically.
    \item We assess the performance of the proposed algorithms both in theory and through simulations. The theoretical analysis comparing the non-private and private methods gives practitioners a sense of how the algorithms would work prior to applying them to real data.
    \item To support the theoretical analysis of one of the algorithms, we study the behavior of a reciprocal normal variable in depth. A general form of the Taylor expansion (for conditional moments) is obtained to solve the problem of the non-existence of moments due to its heavy-tailed nature.
\end{itemize}

The paper is organized as follows. We briefly discuss the existing work on differentially private confidence intervals
in Section 1.1.
Section 2 provides preliminaries on confidence intervals of population proportions and differential privacy. In Section 3, we discuss the methodology of three differentially private algorithms. Section 4 provides theorems on both privacy and asymptotic coverage guarantees. Numerical experiments, including simulation studies and two applications to the 1940 Census data, are conducted in Section 5. Section 6 discusses the implications of
our methods and general research directions on differentially private confidence intervals. 
\subsection{Related Work}
Differentially private confidence intervals
have recently been studied for other settings. 
Some studied differentially private confidence intervals for the population mean of normally distributed data \cite{Karwa2017finite, Du2020DifferentiallyPC, Gaboardi2019LocallyPM} . 
Other tasks on confidence intervals have also been explored.
Drechsler et al. designed and evaluated several strategies to obtain differentially private confidence intervals for the median \cite{Drechsler2021NonparametricDP}.
Wang et al. provided confidence intervals for differentially private models trained with objective or output
perturbation algorithms \cite{Wang2019DifferentiallyPC}.

Besides, bootstrapping is a popular technique for constructing more general differentially private confidence intervals.
Ferrando et al. proposed a general-purpose approach to construct confidence intervals for a population parameter \cite{Ferrando2022ParametricBF}.
A numerical confidence interval for the difference of mean was provided \cite{DOrazio2015DifferentialPF}.
The nonparametric bootstrap was considered 
in \cite{Brawner2018BootstrapIA}.
Covington et al. described a method to induce distributions of mean and covariance estimates via the bag of little bootstraps (BLB), which can further produce private confidence intervals \cite{Covington2021UnbiasedSE}.

Our work is the first to study design-based approaches to sampling.
In a design-based setting, the values of interest are viewed as fixed but unknown constants. Randomness only comes from the sampling design. The selection probabilities introduced with the design will be used for estimation. 
On the contrary, in a model-based setting, a parametric model is postulated.
In many cases, especially with natural populations where no accurate prior information about the population distribution is available, design-based sampling methods can be more reassuring.
More discussion of design-based versus model-based approaches in sampling can be found in \cite{thompson2013theory}.

\section{Preliminaries}
In this section, we provide some preliminaries on population proportion estimation and
differential privacy. We first review the classic Wald confidence interval for the population proportion under stratified random sampling. Then we define a notion of differential privacy specifically for stratified data. Some properties of differential privacy are revisited in preparation for the theoretical analysis in Section 4.
\subsection{Confidence Intervals for the Population Proportion}
In stratified random sampling, a population of $N$ individuals is partitioned into $H$ strata, where stratum $h$ has $N_h$ individuals, and simple random sampling of $n_h$ individuals is conducted within each stratum. When the objective is to estimate the proportion of individuals having some attribute in the population, one can estimate it by
the sample proportion.
Let $y_{hi}$ be the corresponding indicator variable: $y_{hi} = 1$ when the individual $i$ in stratum $h$ has the attribute and $y_{hi} = 0$ otherwise. 
One can estimate the population proportion 
\begin{equation*}
    p = \frac{1}{N}\sum_{h=1}^H \sum_{i=1}^{N_h} y_{hi}
\end{equation*}
 by the sample proportion 
 \begin{equation*}
    \hat p = \frac{1}{N} \sum_{h=1}^H \frac{N_h}{n_h} \sum_{i=1}^{n_h} y_{hi} = \sum_{h=1}^H w_h \hat p_h  
\end{equation*}
where $\displaystyle w_h \stackrel{def}{=} \frac{N_h}{N}$ and $\displaystyle \hat p_h \stackrel{def}{=} \frac{1}{n_h}\sum_{i=1}^{n_h} y_{hi}$. Its variance
$\displaystyle
    {\operatorname{\operatorname{Var}}}(\hat p)= \sum_{h=1}^H w_h^2 { \operatorname{\operatorname{Var}}}(\hat  p_h),
$
where 
\begin{equation*}
    {\operatorname{\operatorname{Var}}}(\hat p_h)= \left( \frac{N_h-n_h}{N_h-1} \right) \frac{ p_h(1- p_h)}{n_h}.
\end{equation*}
An unbiased estimator for ${\operatorname{\operatorname{Var}}}(\hat p_h)$ is given by the sample variance in the stratum
\begin{equation}
    \widehat{ \operatorname{\operatorname{Var}}}(\hat p_h) = \left( \frac{N_h-n_h}{N_h} \right) \frac{\hat p_h(1-\hat p_h)}{n_h-1}.
    \label{eq:varhatp}
\end{equation}
Then an unbiased estimator for ${\operatorname{\operatorname{Var}}}(\hat p)$ is given by
$\displaystyle \widehat{ \operatorname{\operatorname{Var}}}(\hat p)= \sum_{h=1}^H w_h^2 \widehat{ \operatorname{\operatorname{Var}}}(\hat p_h). $
An approximate $100\%(1-\alpha)$ confidence interval for $p$ based on a normal distribution can be constructed:
\begin{equation}
    \hat p \pm z_{1-\frac{\alpha}{2}}\sqrt{{ \operatorname{\widehat{\operatorname{Var}}}}(\hat p)},
\label{ci:wald}
\end{equation}
where $z_{1-\frac{\alpha}{2}}$ denotes the ${1-\frac{\alpha}{2}}$ quantile of standard normal distribution. 
The normal approximation is useful when all the sample sizes are moderate to large. Otherwise, the $t$ distribution with appropriate degrees of freedom is typically used to replace the standard normal distribution.
For small sample sizes, various specialized confidence intervals have been developed \cite{Franco2019ComparativeSO}.

\subsection{Differential Privacy}

Differential privacy ensures that the output of data analysis or a query does not differ much when the data set is changed by one record, such that one can not infer the presence or absence of any individual. 
If two data sets $\bm x, \bm x'$ differ
by one record,
we say that $\bm x, \bm x'$ are \emph{adjacent} or \emph{neighboring}, written as $\bm x \sim \bm x'$.
The definition of differential privacy depends on how we define adjacency.
For the partitioned data under stratified sampling, there are two ways to change a record: (1)
one way is to substitute one record within a stratum, with all the stratum sample sizes fixed. We refer to this adjacency relation as ``\textit{substitute-one relation within a stratum}'' and denote it by $\sim_{ss}$. This relation corresponds to the case where the sample sizes are public and fixed; (2) another way to obtain an adjacent data set is to remove or add one record from one stratum; we refer to the corresponding relation as, which we call ``\textit{remove/add-one relation}'', denoted by $\sim_{r}$. In this case, one of the stratum sample sizes will change by one, as will the overall sample size. This relation corresponds to the case where the sample sizes are private.

Under either adjacency relation, we can define \textit{zero-concentrated differentially private} ($\rho$-zCDP) as in \cite{Bun2016ConcentratedDP}:

\begin{definition}[$\rho$-zCDP] 
Let ${\cal X}^*$ denote the space of the input data with an arbitrary finite dimension. 
Under the adjacency relation $\sim$,
a randomized algorithm $M: {\cal X}^* \rightarrow {\cal Y}$ is $\rho$-zero-concentrated-differentially private ($\rho$-zCDP) if, for every pair of adjacent data sets $\bm x \sim \bm x' \in {\cal X}^*$, and all $\alpha \in (1, \infty)$,
 \begin{equation*}
    \operatorname{D}_\alpha (M(\bm x) \|M(\bm x')) \leq \rho\alpha,
    \label{defzcdp}
 \end{equation*}
where $ \operatorname{D}_\alpha (M(\bm x) \|M(\bm x'))$ is the $\alpha$-Rényi divergence \cite{van2014renyi}
between the distribution of $M(\bm x)$ and
the distribution of $M(\bm x')$. 
\label{def:rhozcdp}
\end{definition}

The parameter $\rho$ indicates the \textit{privacy level}.  A smaller $\rho$ means a more restrictive distance control between $M(\bm x)$ and $M(\bm x')$. As a result, the outputs on two adjacent data sets are harder to tell apart and the algorithm achieves higher privacy. We call $\rho$ the \textit{privacy budget} when we deliberately design an algorithm to satisfy $\rho$-zCDP.

Depending on the adjacency notion, there are two types of differential privacy: \emph{bounded} and \emph{unbounded differential privacy} \cite{Kifer2011NoFL}.
Definition \ref{def:rhozcdp} under the ``remove/add-one relation'' corresponds to the standard unbounded differential privacy. The sample size of the data set changes when one record is added or removed to obtain an adjacent data set.
With ``substitute-one within a stratum'' relation $\sim_{ss}$, the resulting notion corresponds to the bounded version of differential privacy where the sizes of two adjacent data sets are the same. But it is somewhat different from the standard notion of bounded differential privacy in that for the latter, substitutions can happen across strata. That is, we can change both the record and the
stratum it is part of.

In the literature on differential privacy, $(\epsilon, \delta)$-DP (\cite{dwork2014algorithmic} Definition 2.4) is considered the classic notion. We consider $\rho$-zCDP because (1) $\rho$-zCDP implies $(\epsilon, \delta)$-DP (\cite{Bun2016ConcentratedDP} 
 Proposition 1.3), (2) the application of the Gaussian mechanism to achieve zCDP facilitates the theoretical analyses, and (3) the composition of $\rho$-zCDP is straightforward. The Gaussian mechanism is a prototypical example of a mechanism satisfying zCDP, which perturbs the true values by adding Gaussian noise. We provide the Gaussian mechanism and the composition and post-processing properties of $\rho$-zCDP in the following propositions. 
All propositions can be found in \cite{Bun2016ConcentratedDP} and will be used in the analyses of privacy guarantees in Section 4.

\begin{definition}[Sensitivity]
    A function $q$: ${\cal X}^* \to \mathbb{R}$ has sensitivity $\Delta$ if for all pairs of adjacent data sets $x \sim x' \in {\cal X}^* $, we have $|q(x)- q(x')| \leq \Delta$.
\end{definition}

\begin{prop}[Gaussian Mechanism of $\rho$-zCDP]
Let $q: {\cal X}^* \rightarrow \mathbb{R}$ be a sensitivity-$\Delta$ query. Consider the mechanism $M: {\cal X}^* \rightarrow \mathbb{R}$ that on input $x$, releases a sample from $N(q(x), \Delta^2/(2\rho))$. Then, $M$ satisfies $\rho$-zCDP.
\end{prop}
 A smaller budget leads to larger noise added to the query on average. Consequently, the output is more private. 

\begin{prop}[Composition]
Let $M: {\cal X}^* \rightarrow {\cal Y}$ and $M': {\cal X}^* \rightarrow {\cal Z}$ be two randomized algorithms. Suppose $M$ satisfies $\rho$-zCDP and $M'$ satisfies $\rho'$-zCDP, then algorithm $M'' = (M, M'): {\cal X}^* \rightarrow {\cal Y}\times  {\cal Z}$ is $(\rho + \rho')$-zCDP.
\label{prop:composition}
\end{prop}

\begin{prop}[Post-processing]
Let $M: {\cal X}^* \rightarrow {\cal Y}$ and $f: {\cal Y} \rightarrow {\cal Z}$ be randomized algorithms. If $M$ is $\rho$-zCDP, then so is the composed algorithm $M' = f\circ M: {\cal X}^* \rightarrow {\cal Z}$.
\end{prop}

\section{Methodology}
Our goal is to release a $\rho$-zCDP confidence interval for the population proportion $p$ under stratified random sampling. To construct a confidence interval as in (\ref{ci:wald}), we need to estimate both $p$ and the variance of the estimator privately. 
Recall that the non-private estimator of population proportion is given by the sample proportion
\begin{equation*}
    \hat p = \sum_{h = 1}^H w_h\hat p_h.
\end{equation*}
We assume the stratum sizes $N_h$ are all public and fixed, thus so are $w_h$.
To get a private estimator for $p$, denoted by $\tilde p$, we can add noise at the level of either the (non-private) estimator $\hat p$ or the estimator $\hat p_h$. With $\tilde p$, we further devise a private estimator for ${\operatorname{Var}}(\tilde p)$.
Based on this idea, two algorithms for the case of public sample sizes are designed by 
adding noise at the stratum or population level in section 3.1. 
In section 3.2, we additionally propose adding noise at the stratum level when sample sizes are private.  
Throughout the paper,
the accents $\bm\hat{\cdot}$ and $\bm \tilde{\cdot}$ are used to represent non-private and private estimators, respectively.

\subsection{Estimating with Public Sample Sizes}
When sample sizes $n_h$ are fixed, there are two natural strategies for perturbing $\hat p$:
add Gaussian noise to (1) the stratum-level statistics $\hat p_h$'s, or (2) the overall statistic $\hat p$. Adding noise to the $\hat p_h$'s has the advantage of producing private estimators for stratum-level proportions simultaneously.

\subsubsection{Adding Noise at the Stratum Level}

We apply the Gaussian mechanism to each stratum to derive a private estimator $\tilde p_h \stackrel{def}{=} \hat p_h + e_h$ where $e_h$ is the Gaussian noise.
Then the private estimator for the population proportion is
\begin{equation*}
    \tilde p \stackrel{def}{=}\sum_{h=1}^H w_h\tilde{p}_h.
\end{equation*}
As a result, the variance of $\tilde{p}$ consists of both the intrinsic variances of estimating $p_h$'s by $\hat p_h$'s and the additional variability from added noise:
\begin{equation}
    {\operatorname{Var}}(\tilde p) =\sum_{h=1}^H w_h^2 \left( \operatorname{Var}(\hat p_h) + w_h^2 \operatorname{Var}(e_h) \right)
    \label{eq:true var stra}
\end{equation}
where $\operatorname{Var}(e_h), h = 1, ..., H$ are public since they do not depend on the data. 

To obtain a private confidence interval for $\hat p$,
we will need to privately estimate ${\operatorname{Var}}(\hat p_h)$. Note that the added noise biases the term $\hat p_h(1-\hat p_h)$ in the non-private estimate of ${\operatorname{Var}}(\hat p_h)$ in (\ref{eq:varhatp}). More specifically,
 $\mathbb{E}_{e} [\tilde p_h(1-\tilde p_h)] = \hat p_h(1-\hat p_h) - \operatorname{Var}(e_h)$ where $\mathbb{E}_{e}$ denotes the expectation taken on the randomness of the added noise. 
Then a private unbiased estimator of $\operatorname{Var}(\hat p_h)$ in the right-hand side in (\ref{eq:true var stra}) is given by
\begin{equation}
    \widetilde{\operatorname{Var}}(\hat p_h) \stackrel{def}{=} \left(\frac{N_h-n_h}{N_h} \right)\frac{\tilde p_h(1-\tilde p_h) +\operatorname{Var}(e_h)}{n_h-1}.
\label{eq:nas_var_ph}
\end{equation}
To estimate ${\operatorname{Var}}(\tilde p)$, we set
\begin{equation*}
    \widetilde{\operatorname{Var}}(\tilde p) \stackrel{def}{=}  \sum_{h = 1}^H w_h^2 \left( \widetilde{\operatorname{Var}}(\hat p_h) + \operatorname{Var}(e_h) \right)
\end{equation*}
This approach is formulated in Algorithm \ref{alg:strlevel} which we call \textbf{StrNz-PubSz} (adding noise at the stratum level with public sample sizes). The theoretical results regarding privacy level and asymptotic coverage are provided in Theorems \ref{thm:privacy} and \ref{thm:NoiStraPubSize}.
 
\begin{algorithm}[!ht]
\begin{algorithmic}[1]
\Require $\hat p_h$, $n_h$, $N_h$, $w_h$, $\rho$, $\alpha$.
\Ensure $\rho$-zCDP $(1-\alpha)$ CI  for the population proportion. 
\For{$h=1$ to $H$}
    \State Generate Gaussian noise $e_h \sim \mathcal{N}(0, \frac{1}{2\rho n_h^2})$, and let $$\tilde p_h \leftarrow \hat p_h + e_h.$$
    \State Estimate ${\operatorname{Var}}(\tilde p_h)$ by $$\widetilde{V}_h \leftarrow \left(\frac{N_h-n_h}{N_h}\right)\frac{\tilde p_h(1-\tilde p_h) +\frac{1}{2\rho n_h^2}}{n_h-1} + \frac{1}{2\rho n_h^2}.$$
\EndFor

\State Estimate $p$ by $\displaystyle \tilde p \leftarrow \sum_{h=1}^H w_h\tilde{p}_h$ and ${\operatorname{Var}}(\tilde p)$ by
$\displaystyle \widetilde{V} \leftarrow \sum_{h=1}^H w_h^2 \widetilde{V}_h.$
    \State Return $$\tilde p \pm z_{1-\alpha/2}\sqrt{\widetilde{V}},$$
    where $z_{1-\alpha/2}$ is the $({1-\alpha/2})$-quantile of the standard normal distribution.
\end{algorithmic}
\caption{Adding noise at the stratum level with public sample sizes, {StrNz-PubSz}}
\label{alg:strlevel}
\end{algorithm}

\subsubsection{Adding Noise at the Population Level}
An alternative strategy is to directly add noise to the non-private estimator of $p$, i.e., $\hat p$. The sensitivity of $\hat p$ is $$\Delta_p = \max_h\frac{w_h}{n_h}.$$
Since $w_h$ and $n_h$ are public, $\Delta_p$ can be made public.
We set $\tilde p = \hat p + {e}$ where $e$ is the Gaussian noise with standard deviation proportional to $\Delta_p$.
Then, the variance of $\tilde p$ becomes 
\begin{equation}
    {\operatorname{Var}}(\tilde p) = \operatorname{Var}(\hat p)+\operatorname{Var}(e).
    \label{eq:true var pop}
\end{equation}
Recall that $$\widehat{\operatorname{Var}}(\hat p) = \sum_{h=1}^H w_h^2 \left(\frac{N_h-n_h}{N_h} \right)\frac{\hat p_h(1-\hat p_h) }{n_h-1}$$ is an unbiased estimator for ${\operatorname{Var}}(\hat p)$.
To get a private estimator for ${\operatorname{Var}}(\tilde p)$, we again apply the Gaussian mechanism to $\widehat{\operatorname{Var}}(\hat p)$ based on the sensitivity of $\operatorname{Var}(\hat p)$:
\begin{equation*}
    \Delta_{V} = \max_h\left(\frac{C_h}{n_h}\left(1-\frac{1}{n_h}\right) \right),
\end{equation*}
where $C_h = w_h^2 \frac{N_h-n_h}{N_h}\frac{1}{n_h-1}$.

Since we apply the Gaussian mechanism twice, the total privacy budget should be divided into two parts: $\rho = \rho_1 + \rho_2$ to spend on adding noise to $\hat p$ and $\operatorname{Var}(\hat p)$, respectively.
The composition property (Proposition \ref{prop:composition}) ensures the total privacy level is $\rho$. The resulting algorithm, \textbf{PopNz-PubSz}, is presented in Algorithm \ref{alg:poplevel}.

\begin{algorithm}[!ht]
\begin{algorithmic}[1]
\Require $\hat p$, $\hat p_h$, $n_h$, $N_h$, $w_h$, $\rho$, $\alpha$.
\Ensure A $\rho$-zCDP $(1-\alpha)$ CI for Population Proportion. 
\State Split the budget $\rho = \rho_1 + \rho_2$. 

\State Generate noise $e \sim \mathcal{N}(0, \frac{\Delta_{ p}^2}{2\rho_1})$ where $\Delta_p = \max_h\frac{w_h}{n_h}$ and let $$\tilde p \leftarrow \hat p + e.$$
\State Generate noise $e_{V} \sim \mathcal{N}(0, \frac{\Delta_{V}^2}{2\rho_2})$ where $\Delta_{V} = \max_h\left(\frac{C_h}{n_h}\left(1-\frac{1}{n_h}\right) \right)$ and $C_h = w_h^2 \frac{N_h-n_h}{N_h}\frac{1}{n_h-1}$.  
Let $$\widetilde{V} \leftarrow \sum_{h=1}^H w_h^2 \left( \frac{N_h-n_h}{N_h} \right)\frac{\hat p_h(1-\hat p_h) }{n_h-1} + \frac{\Delta_{ p}^2}{2\rho_1} + e_{V}.$$
\State Return $$\tilde p \pm z_{1-\alpha/2}\sqrt{\widetilde{V}},$$
where $z_{1-\alpha/2}$ is the $({1-\alpha/2})$-quantile of the standard normal distribution.
\end{algorithmic}
\caption{Adding noise at the population level with public sample sizes, {PopNz-PubSz}}
\label{alg:poplevel}
\end{algorithm}

\begin{remark} When there are multiple strata with similar sampling rates, Algorithm \ref{alg:strlevel} yields a wider confidence interval for $p$ than Algorithm \ref{alg:poplevel} does, given the same privacy budget.
However, Algorithm \ref{alg:strlevel} additionally produces private confidence intervals for $\hat p_h$ which may be of interest for release. In Section 4.2.1, we compare the two algorithms quantitatively.
\end{remark}

\begin{remark}
Proportions are always between 0 and 1. One can post-process proportion estimates ($\tilde p_h$ in Algorithm \ref{alg:strlevel} and $\tilde p$ in Algorithm \ref{alg:poplevel}) by clipping them onto interval [0,1] without undermining privacy.
When the privacy budget is very small, the noisy proportion estimates are likely to lie outside [0,1]. Thus, clipping moves the confidence interval toward the truth
and a higher coverage rate will be observed. With a moderate or large budget, clipping does not make a noticeable difference. 

Lastly, one can always clip the output confidence intervals onto [0,1] without privacy loss.
\label{rmk:clipping}
\end{remark}

\subsection{Estimating with Private Sample Sizes}
\label{sub:privatesize}

When sample sizes are public information, keeping the proportions private is essentially protecting only the numerator, i.e., the counts of individuals with $y = 1$.
In some cases where subpopulation proportions also need to be estimated, 
Algorithms \ref{alg:strlevel} and \ref{alg:poplevel} with public sample sizes can lead to privacy leakage since the counts become the denominator. 
For example, one may ask the following queries: (1) what is the proportion of females in the US; and (2) what is the proportion of unemployed among females in the US. The number of females is the numerator in query (1) but becomes the denominator in query (2). Employing Algorithms \ref{alg:strlevel} or \ref{alg:poplevel} protects the number of females in query (1) but reveals it in query (2).
Therefore, a method of constructing confidence intervals for proportions to keep both the counts and sample sizes private is necessary for subpopulation analysis.  
We protect the sample sizes by adding noise to them. As a result, sample sizes are not fixed and therefore we need the unbounded notion of differential privacy with the adjacency relation $\sim_r$.

In the following, we extend Algorithm \ref{alg:strlevel} to serve the needs of privacy protection of sample sizes by adding noise at the stratum level.
(It is not obvious how to extend Algorithm \ref{alg:poplevel}, which adds noise at the population level.
It requires more sophisticated mechanisms; we briefly discuss in Section 6.)

To begin, we first consider the setting of simple random sampling. The idea is to add independent Gaussian noise to both the numerator and denominator for each stratum. For ease of notation, we first consider a single stratum with count $c = \sum_{i=1}^nx_i$. We know
\begin{equation*}
    c \sim \text{Hypergeometric}(N, K, n),
\end{equation*}
where $K$ is the total number of individuals with the attribute of interest. 
The count $c$ has mean $n\frac{K}{N} = np$ and variance $ n\frac{K}{N}\frac{N-K}{N}\frac{N-n}{N-1} = n^2\operatorname{Var}(\hat p)$.
By applying the Gaussian mechanism to $c$ and $n$ with privacy budgets $\rho_1$ and $\rho_2$, respectively, we have private count $\tilde c$ and sample size $\tilde n$:
\begin{equation*}
    \tilde c \mid c \sim \mathcal{N}(c,  \frac{1}{2\rho_1})
\end{equation*}
and
\begin{equation*}
    \tilde n  \sim \mathcal{N}(n, \frac{1}{2\rho_2}).
\end{equation*}
The unconditional mean and variance for $c$ are
\begin{equation*}
    \mathbb{E}(\tilde c) =  \mathbb{E}[\mathbb{E}(\tilde c \mid c)] = \mathbb{E}( c) = np
\end{equation*}
and
\begin{equation}
\begin{aligned}
    \operatorname{Var}(\tilde c) & = 
    \mathbb{E}\operatorname{Var}(\tilde c \mid c)+
    \operatorname{Var}\mathbb{E} (\tilde c \mid c) = \frac{1}{2\rho_1} + n^2 \operatorname{Var}(\hat p).
\end{aligned}
\label{eq:varofc}
\end{equation}
By the composition property of zCDP, we get a private estimator for proportion $p$, denoted by $\tilde p$, with privacy level $\rho = \rho_1 + \rho_2$. 
Since $\tilde c $ and $ \tilde n$ are independent variables, in principle,
\begin{equation}
    \mathbb{E}(\tilde p) =  \mathbb{E}\left(\frac{\tilde c}{\tilde n} \right) = \mathbb{E}(\tilde c) \mathbb{E}\left(\frac{1}{\tilde n} \right),
\label{eq: mean}
\end{equation}
and
\begin{equation}
\begin{aligned}
    \operatorname{Var}(\tilde p) & = \mathbb{E}\left(\frac{\tilde c}{\tilde n} \right)^2 - \left(\mathbb{E}\left(\frac{\tilde c}{\tilde n} \right)\right)^2  = \mathbb{E}\tilde c^2\mathbb{E}\left(\frac{1}{\tilde n^2}\right) - (\mathbb{E}\tilde c)^2\left(\mathbb{E}\frac{1}{\tilde n}\right)^2.
\end{aligned}
\label{eq:var}
\end{equation}

However, the moments of $\frac{1}{\tilde n}$ do not exist, thus neither do those of $\tilde p$.
Generally speaking, the ratio of two independent normal random variables has a heavy-tailed distribution with no moments \cite{Marsaglia2006ratios, ratiov2013onthe}. The shape of the distribution could be unimodal, bimodal, symmetric, or asymmetric. It is primarily determined by the coefficient of variation 
of the denominator variable, $CV$. When $CV$ is sufficiently small, a normal distribution approximation can be effective. 
It has been shown theoretically that a normal distribution can be arbitrarily close to the ratio variable in an interval centered at the ratio of means of two normal random variables \cite{ratiov2013onthe}.
Experiments have provided guidelines for when the normal approximation can be used.
For example, a simple rule is that the approximation is reasonable whenever $CV$ is less than 0.1 \cite{kuethe2000Imaging}. 
Other practical rules are mentioned in \cite{Hayya1975anote, Marsaglia2006ratios}.

\begin{algorithm}[!ht]
\begin{algorithmic}[1]
\Require $N_h$, $w_h$,  $n_h$, $c_h$, $\rho$, $\alpha$.
\Ensure A $\rho$-zCDP $(1-\alpha)$ CI  for the population proportion. 
\State Split the budget 
$\rho = \rho_1 + \rho_2$. 
\For{$h=1$ to $H$}
    \State Generate $e_{h}^{(1)} \sim \mathcal{N}(0, \frac{1}{2\rho_1})$ and $e_{h}^{(2)} \sim \mathcal{N}(0, \frac{1}{2\rho_2})$, and let
    \begin{equation}
        \begin{cases}
        \tilde c_h \leftarrow c_h + e_{h}^{(1)}\\
        \tilde n_h \leftarrow \max(n_h + e_{h}^{(2)}, 2)
    \end{cases}
    \label{eq:tilde n_h}
    \end{equation}

    \State Let 
    \begin{equation}
    \tilde p_h \leftarrow \frac{\tilde c_h}{\tilde n_h}
        \label{eq:tilde p_h}
    \end{equation}
    \State Let 
    \begin{equation}
        \widetilde V_h \leftarrow  \left(\frac{N_h-\tilde n_h}{N_h - 1}\right) \frac{\tilde p_h(1- \tilde p_h)}{\tilde n_h} + \frac{1}{2\rho_1 \tilde n_h^2} +  \frac{\tilde p_h^2}{2\rho_2 \tilde n_h^2}.
        \label{eq:tilde v_h}
    \end{equation}
\EndFor
\State Estimate $p$ by $\displaystyle \tilde p \leftarrow \sum_{h=1}^H w_h\tilde{p}_h$ and let $\displaystyle \widetilde{V} \leftarrow \sum_{h=1}^H w_h^2 \widetilde{V}_h.$

    \State Return $$\tilde p \pm z_{1-\alpha/2}\sqrt{\widetilde V},$$ where $z_{1-\alpha/2}$ is the $({1-\alpha/2})$-quantile of the standard normal distribution.
\end{algorithmic}
\caption{Adding noise at the stratum level with private sample sizes, {StrNz-PrivSz}}
\label{alg:{StrNz-PrivSz}}
\end{algorithm}

We take advantage of the normal approximation to construct a $\rho$-zCDP confidence interval for the proportion.
We present the following estimation strategy in Algorithm \ref{alg:{StrNz-PrivSz}}, \textbf{StrNz-PrivSz}. In the algorithm, we clip $\tilde n_h$ in (\ref{eq:tilde n_h}) to ensure the denominator is not too small. Otherwise, the ratio can be arbitrarily large. Such a post-processing step preserves the same privacy guarantee. 
For the theoretical analysis, we do not clip $ \tilde n_h$, but instead, we consider the ratio variable $\tilde c_h/\tilde n_h$ given the event $S_h = \{1\leq  \tilde n_h\leq 2n_h-1\}$ (a symmetric area around the mean of $\tilde n_h$). It is more convenient for the analysis.
The asymptotic behaviors of $\tilde p_h$ in the algorithm and $\tilde c_h/\tilde n_h \mid S_h$ are essentially the same since $\Pr(\tilde n_h \geq 2) \rightarrow 1$ and $\Pr(S_h) \rightarrow 1$ as $n \rightarrow \infty$. 
We will see the private estimator of the variance of $\tilde p_h$ we derive from the analysis of $\tilde c_h/\tilde n_h \mid S_h$ works well and the algorithm does achieve the desired coverage level. 

We consider the ratio of two independent normal variables. By independence, what remains unclear is the behavior of the reciprocal of a normal distribution. (We should mention that the Inverse Gaussian distribution is a different distribution than the reciprocal distribution we discuss here.) 
In Theorem \ref{prop:condmom}, we provide a general form of the Taylor series of conditional mean and variance of a reciprocal normal distribution. 
To our best knowledge, this is the first complete result of the Taylor series, with the remainder term specified.
We prove the theorem in the Proofs section.   
We use $k = 2$ to derive an estimator for the variance of $\tilde p$ Algorithm \ref{alg:{StrNz-PrivSz}}, which leads to (\ref{eq:tilde v_h}).

\begin{theorem}[Conditional mean and variance of a reciprocal normal distribution]
For random variable $X\sim \mathcal{N}(\mu, \sigma^2)$ where $\mu > 1$ and $\sigma^2 > 0$, given the event $S = \{1\leq  X \leq 2\mu-1\}$, for any integer $k>0$, the first two moments of $\frac{1}{X} \mid S$ have the following expansions:
\begin{equation}
\begin{aligned}
        \mathbb{E}\left( \frac{1}{X} \mid S \right) 
        = \frac{1}{\mu}\sum_{j = 0}^{k} \frac{(2j-1)!!\sigma^{2j}}{\mu^{2j}}   + O\left(\frac{\sigma^{2k+2}}{\mu^{2k+2}}\right)
\end{aligned}
\label{eq:condmean}
\end{equation}
and 
\begin{equation}
\begin{aligned}
       \mathbb{E}\left( \frac{1}{X^2} \mid S \right)
        & = \frac{1}{\mu^2}\sum_{j = 0}^{k} \frac{(2j+1)!!\sigma^{2j}}{\mu^{2j}}   + O\left(\frac{\sigma^{2k+2}}{\mu^{2k+2}}\right).
\end{aligned}
\label{eq:condvar}
\end{equation}
\label{prop:condmom}
\end{theorem}

\section{Theoretical Results}
In this section, we present the theoretical results of both privacy and asymptotic coverage guarantees. In addition, comparisons of the three algorithms in terms of variance and width ratios are discussed.

\subsection{Privacy and Coverage Guarantees}
\label{sec:thmresults}
Our theoretical results are two-fold. First, the proposed algorithms satisfy the desired privacy level under the 
corresponding adjacency relation, which is presented in Theorem \ref{thm:privacy}.
\begin{theorem}[Privacy Guarantee]
Algorithms \ref{alg:strlevel} and \ref{alg:poplevel} satisfy $\rho$-zCDP under the adjacency relation $\sim_{ss}$; Algorithm \ref{alg:{StrNz-PrivSz}} satisfies $\rho$-zCDP under the adjacency relation $\sim_{r}$.
\label{thm:privacy}
\end{theorem}
Proofs are presented in the Proofs section. 

On the other hand, for the confidence intervals to be useful, we provide theorems that guarantee the asymptotic coverage for each algorithm.
The central limit theorem (CLT) asserts (essentially) that the sample mean is asymptotically normally distributed regardless of the original distribution. Therefore, the sample mean can be used to construct a confidence interval for the population mean.
In the finite-population sampling designs we are considering, variants of CLTs can be found among 
\cite{Erdos2004ONTC, hajek:1959, Lehmann} and others. We restate a general form of the finite-population CLT for simple random sampling in Theorem \ref{thm: SRS_CLT} and provide asymptotic coverage guarantees in the following theorems. 

\begin{theorem}[Algorithm \ref{alg:strlevel}]
For a population of size $N$, let $p$ be the proportion in the population with the attribute of interest. 
Under stratified random sampling with sample sizes $n_h$ within the stratum of size $N_h$, $h = 1,.., H$, let $\displaystyle\widetilde V = \sum_{h=1}^H w_h^2\widetilde V_h$ where 
\begin{equation}
    \widetilde{V}_h = \left(\frac{N_h-n_h}{N_h}\right)\frac{\tilde p_h(1-\tilde p_h) +\frac{1}{2\rho n_h^2}}{n_h-1} + \frac{1}{2\rho n_h^2}.
\end{equation}
for $\rho>0$ as described in Algorithm \ref{alg:strlevel}.
If $\rho = \omega(1/n_h)$ for all $h$, then as $N_h - n_h$ and $n_h$ both tend to infinity for every stratum,

\begin{enumerate}[label=(\roman*)]
    \item
     $\widetilde{V} \stackrel{p}{\rightarrow} \operatorname{Var}(\tilde p)$, 
    more specifically,  for all $h$,
    \begin{equation*}
        \widetilde{V}_h-\operatorname{Var}(\tilde p_h) = \widehat{\operatorname{Var}}(\hat p_h) -  \operatorname{Var}(\hat p_h) + O_P\left(\frac{1}{\sqrt{\rho}n_h^2}\right) = O_P\left(\frac{1}{n_h^{3/2}}\right),
    \end{equation*}
    where $\widehat{\operatorname{Var}}(\hat p_h)$ is the non-private estimator for $\operatorname{Var}(\hat p_h)$;

    \medskip
    \item for $0 <\alpha < 1$, 
    \begin{equation}
    \Pr \left( p \in \left(\tilde p - z_{1-\alpha/2}\sqrt{\widetilde V}, \tilde p + z_{1-\alpha/2}\sqrt{\widetilde V}\right)\right)  \to 1 - \alpha.
    \end{equation}

\end{enumerate}

\label{thm:NoiStraPubSize}
\end{theorem}

\begin{theorem}[Algorithm \ref{alg:poplevel}] 
For a population of size $N$, let $p$ be the proportion in the population with the attribute of interest. 
Under stratified random sampling with sample sizes $n_h$ within the stratum of size $N_h$, $h = 1,.., H$, let 
\begin{equation}
    \widetilde{V} = \sum_{h=1}^H w_h^2 \left( \frac{N_h-n_h}{N_h} \right)\frac{\hat p_h(1-\hat p_h) }{n_h-1} + \frac{\Delta_{ p}^2}{2\rho_1} + e_{V}
\end{equation}
where $e_{V} \sim \mathcal{N}(0, \frac{\Delta_{V}^2}{2\rho_2})$
for $\rho_1,\rho_2 > 0$ as described in Algorithm \ref{alg:poplevel}.
If $\rho_1 = \omega(1/n_h)$ and $\rho_2 = \omega(1/n_h)$ for all $h$, then as $N_h - n_h$ and $n_h$ both tend to infinity for every stratum, 

\begin{enumerate}[label=(\roman*)]
        \item
        $\widetilde{V} \stackrel{p}{\rightarrow} \operatorname{Var}(\tilde p)$, 
        more specifically, 
        \begin{equation*}
            \widetilde{V}-\operatorname{Var}(\tilde p) = \widehat{\operatorname{Var}}(\hat p) -  \operatorname{Var}(\hat p) + O_P\left(\frac{1}{\max\limits_{h}\sqrt{\rho_2}n_h^2}\right) = O_P\left(\frac{1}{\max\limits_h n_h^{3/2}}\right);
        \end{equation*}
        where $\widehat{\operatorname{Var}}(\hat p)$ is the non-private estimator for $\operatorname{Var}(\hat p)$;
        \medskip
        \item for $0 <\alpha < 1$, 
    \begin{equation}
    \Pr \left( p \in \left(\tilde p - z_{1-\alpha/2}\sqrt{\widetilde V}, \tilde p + z_{1-\alpha/2}\sqrt{\widetilde V}\right)\right)  \to 1 - \alpha.
    \end{equation}
\end{enumerate}
\label{thm:NoiPopPubSize}
\end{theorem}

Proofs of the above theorems use the finite-population CLT and are provided in the Proofs section.

For Algorithm \ref{alg:{StrNz-PrivSz}}, the asymptotic behavior of $\tilde p$ is grounded on the normal approximation to a ratio variable in addition to the CLT. 
We revisit the result of normal approximation by \cite{ratiov2013onthe} in Theorem \ref{lem:ratio}. 
Based on the approximation, 
we have shown the consistency of $\tilde p$ in the case of simple random sampling.
\begin{theorem} Under simple random sampling, let $c$ be the count of individuals having the attribute of interest and $n$ be the sample size. The true population proportion is denoted by $p$. Let $\tilde p = \tilde c / \tilde n$ where $\tilde c \sim \mathcal{N}(c,  \frac{1}{2\rho_1})$ and  $\tilde n \sim \mathcal{N}(n,  \frac{1}{2\rho_2})$ for $\rho_1,\rho_2 >0$.
Under the conditions that $\rho_2 = \omega(1/n)$, $\rho_1 = \omega(1/n)$, $\tilde p$ is a consistent estimator for $p$.
\label{thm:consistency}
\end{theorem}

With the foundation of the above consistency, we establish the asymptotic properties:
\begin{theorem}[Algorithm \ref{alg:{StrNz-PrivSz}}] 
For a population of size $N$, let $p$ be the proportion in the population with the attribute of interest. 
Under stratified random sampling with sample sizes $n_h$ within the stratum of size $N_h$, $h = 1,.., H$, 
let $\displaystyle \widetilde V = \sum_{h=1}^H w_h^2\widetilde V_h$ where
\begin{equation}
    \widetilde V_h = \left(\frac{N_h-\tilde n_h}{N_h - 1}\right) \frac{\tilde p_h(1- \tilde p_h)}{\tilde n_h} + \frac{1}{2\rho_1 \tilde n_h^2} +  \frac{\tilde p_h^2}{2\rho_2 \tilde n_h^2}
\end{equation}
for $\rho_1,\rho_2 >0$ as described in Algorithm \ref{alg:{StrNz-PrivSz}}.
If $\rho_1 = \omega(1/n_h)$ and $\rho_2 = \omega(1/n_h)$ for all $h$, then as $N_h - n_h$ and $n_h$ both tend to infinity for every stratum, 
\begin{enumerate}[label=(\roman*)]
    \item
     $\widetilde{V} \stackrel{p}{\rightarrow} \operatorname{Var}(\tilde p\mid S)$ where $S$ is an event with $\Pr(S) \to 1$,  more specifically, for all $h$,
    \begin{equation*}
        \tilde{V}_h - \operatorname{Var}(\tilde p_h \mid S) = \widehat{\operatorname{Var}}(\hat p_h) -  \operatorname{Var}(\hat p_h) + 
        O_p\left(\frac{1}{\rho_1n_h^2} + \frac{1}{\rho_2n_h^2}\right)= o_p\left(\frac{1}{n_h}\right),
    \end{equation*}
     where $S_h$ is an event with $\Pr(S_h) \to 1$;
    
    \medskip
    \item
    for $0 <\alpha < 1$, 
    \begin{equation}
    \Pr \left( p \in \left(\tilde p - z_{1-\alpha/2}\sqrt{\widetilde V}, \tilde p + z_{1-\alpha/2}\sqrt{\widetilde V}\right)\right)  \to 1 - \alpha.
    \end{equation}

\end{enumerate}
\label{thm:NoiStraPrivSize}
\end{theorem}

The event $S_h$ is discussed in Section \ref{sub:privatesize} and $S$ can be set to $\cap_h S_h$.
As the variance of both $\tilde p$ and $\tilde p_h$ does not exist, we resort to the conditional variance
under high probability events.
To prove Theorem \ref{thm:NoiStraPrivSize}, we start with a single stratum. We use a normal distribution (denoted by $p_h^*$) to approximate that of the proportion estimator $\tilde p_h$, with the distance between the two distribution vanishing to zero in an interval.
Then for multiple strata, we show that the linear combination of the normal variables (denoted by $p^*$) is  an 
 accurate 
 approximation to $\tilde p$. Last but not least, due to the consistency stated in Theorem \ref{thm:consistency}, the noisy estimator $\widetilde V$ is a consistent estimator for the variance of $p^*$. 
 Then, a Wald confidence interval can be constructed using $\tilde p$ and $\widetilde V$. Details are presented
in the Proofs section.

Note that, in addition to consistency for our estimates of the variance, the results above provide convergence rates. 
Compared to the estimation of variance in non-private settings, the additional biases are merely nuances given the conditions on $\rho$, $\rho_1$ and $\rho_2$.
In fact, we impose these conditions to ensure that the introduced noise does not dominate when estimating the variance.
In principle, these rates may be used in practice to adjust the length of confidence intervals accordingly, although we do not explore that direction here.

\subsection{Comparisons of Variances}

The theorems presented in Section 4.1 ensure that, under proper conditions, the desired coverage is achieved asymptotically. 
Therefore, to compare the performance of the different proposed confidence intervals, we compare their widths, which are determined by their variance estimates. In this section, we will analyze our variance estimates and compare the resulting widths to that of the non-private confidence interval.

\subsubsection{Extrinsic Variances}

To investigate how much additional uncertainty is caused by adding noise, we decompose the variances of the private estimators into two parts: (1) the inherent component coming from the estimation from the sampling data, i.e, $\operatorname{Var}(\hat p)$, and (2) the extrinsic component introduced by the added noise, written as $$V_{\text{ex}} \stackrel{def}{=} \operatorname{Var}(\tilde p) - \operatorname{Var}(\hat p).$$
Table \ref{tab:vex} provides the (approximate) variances of $\tilde p$ for three algorithms, where $w_h = \frac{N_h}{N}$ are the stratum weights. The variances are derived in the proofs of Theorems \ref{thm:NoiStraPubSize}, \ref{thm:NoiPopPubSize}, and \ref{thm:NoiStraPrivSize}. The additional variance terms, $V_{\text{ex}}$, can be rewritten in terms of $u_h \stackrel{def}{=} \frac{N_h}{n_h}$ instead of $w_h$, as shown in the second row of the table.
In fact, $u_h$ are called \textit{sampling weights} in the literature on survey sampling. A sample weight is defined as the number of individuals that each respondent in the sample is representing in the population.  It is the reciprocal of the sampling rate $\frac{n_h}{N_h}$ and
plays an important role in statistical inference for survey data \cite{Pfeffermann1993therole, DuMouchel1983using}. Understanding the relation between sampling weights and the variance of the noisy estimators is helpful for practitioners to make survey designs and the choice of algorithms.

\begin{table}[H]
\caption{(Approximate) variances of $\tilde p$. }
\setlength{\tabcolsep}{8pt}
\renewcommand{\arraystretch}{1.8}
\resizebox{\columnwidth}{!}{%
\begin{tabular}{@{}c|ccc@{}}
\toprule
Algorithm                      & {StrNz-PubSz}                                                                 & {PopNz-PubSz}                                                              & {StrNz-PrivSz} (approximate)                                                                                                                                \\ \midrule
$\operatorname{Var}(\tilde p)$& $\operatorname{Var}(\hat p) + \frac{1}{2\rho} \sum_{h = 1}^H \frac{w_h^2}{n_h^2}$ & $\operatorname{Var}(\hat p) + \frac{1}{2\rho_1} \max_{h} \frac{w_h^2}{n_h^2}$ & $\operatorname{Var}(\hat p) + \frac{1}{2\rho_1}\sum_{h = 1}^H \frac{w_h^2}{n_h^2} + \frac{1}{2\rho_2}\sum_{h = 1}^H \frac{w_h^2 p_h^2}{n_h^2}$ \\
$V_{\text{ex}}$                & $\frac{1}{2N^2}\sum_{h=1}^H\frac{ u^2_h}{\rho} $                                            & $\frac{1}{2N^2}\max_h\frac{ u^2_h}{\rho_1} $                                            & $\frac{1}{2N^2}\sum_{h=1}^H u^2_h (\frac{1}{\rho_1} + \frac{  p_h^2}{\rho_2})$      \\
\bottomrule
\end{tabular}}

\label{tab:vex}
\end{table}
With a fixed population size $N$ and a chosen privacy level $\rho$, the extra variances $V_{\text{ex}}$ induced by the added noise are primarily dictated by $u_h$.
In {PopNz-PubSz} where we add noise at the population level, $V_{\text{ex}}$ is solely determined by the largest sample weight among all strata. If noise is injected into each stratum, then sampling weights in all strata collectively affect $V_{\text{ex}}$.
In particular, for {StrNz-PrivSz}, $V_{\text{ex}}$ is impacted by $p_h$ additionally. 
For all three algorithms, smaller sampling weights lead to lower extrinsic variance. 

For comparison, we look at the ratio of  $V_{\text{ex}}$ with the budgeting $\rho_1 = \rho_2 = \rho/2$ for {PopNz-PubSz} and {StrNz-PrivSz}. The ratio of $V_{\text{ex}}$ for {StrNz-PubSz} to {PopNz-PubSz} is 
\begin{equation}
    \frac{\sum_{h = 1}^H u_h^2}{2 \max_h u_h^2}.
    \label{eq:weight1}
\end{equation}
Roughly speaking, when there are many strata, adding noise at the population level gives a smaller variance.
To compare {StrNz-PrivSz} and {StrNz-PubSz}, the ratio of $V_{\text{ex}}$ is
\begin{equation}
    \frac{2\sum_{h = 1}^H u_h^2(1+ p_h^2)}{\sum_{h = 1}^H u_h^2},
     \label{eq:weight2}
\end{equation}
which will always be greater than 2 (due to the cost it takes to protect sample sizes in {StrNz-PrivSz}) and at most 4.

\subsubsection{Comparing with Non-Private CI: One Stratum Case}
To assess the width in theory,
we also look at the confidence interval width ratios by comparing them to the non-private one. 
Since the parameters $N_h$, $n_h$, $p_h$, $\rho_h$ come into play together in the stratification setting, it is more practical to analyze the special case with one stratum.  

Let the theoretical width ratio (TWR) be
\begin{equation*}
    \text{TWR} = \sqrt{\frac{\operatorname{Var}(\tilde p)}{\operatorname{Var}(\hat p)}}.
\end{equation*}
In the implementation, the real width ratio (WR), defined as $\sqrt{\widetilde V / \operatorname{Var}(\hat p)}$, will be very close to TWR in that $\widetilde V$ is a consistent estimator for $\operatorname{Var}(\tilde p)$. Table \ref{tab:TWR} displays some relevant quantities. 
Note that $\frac{N-1}{N-n}$ is always less than 1 but tends to 1 when the population size is far larger than the sample size. 

\begin{table}[ht]
\caption{Theoretical width ratios and lower bounds. The budgeting $\rho_1 = \rho_2 = \rho/2$ are used for {PopNz-PubSz} and {StrNz-PrivSz}.}
\setlength{\tabcolsep}{8pt}
\renewcommand{\arraystretch}{2}
\resizebox{\columnwidth}{!}{%
\begin{tabular}{@{}c|ccc@{}}
\toprule
Algorithm                      & {StrNz-PubSz}                                  & {PopNz-PubSz}                                  & {StrNz-PrivSz}                                                                 \\ \midrule
$\tilde p$                     & $\hat p + \mathcal{N}(0, \frac{1}{2\rho n^2})$         & $\hat p + \mathcal{N}(0, \frac{1}{\rho  n^2})$      & $(c + \mathcal{N}(0, \frac{1}{\rho})) / (n + \mathcal{N}(0, \frac{1}{\rho}))$ \\
$\operatorname{Var}(\tilde p)$ & $\operatorname{Var}(\hat p) + \frac{1}{2\rho  n^2}$     & $\operatorname{Var}(\hat p) + \frac{1}{\rho  n^2} $ & $\operatorname{Var}(\hat p) +\frac{1 + p^2}{\rho n^2}$               \\
TWR                  & $\sqrt{1 + \frac{N-1}{N-n}\frac{1}{2p(1-p)n\rho}}$ & $\sqrt{1 + \frac{N-1}{N-n}\frac{1}{p(1-p)n\rho}}$ & $\sqrt{1 + \frac{N-1}{N-n}\frac{1+p^2}{p(1-p)n\rho}}$                              \\
Lower bound of TWR                & $\sqrt{1 + \frac{2}{n\rho}}$ & $\sqrt{1 + \frac{4}{n\rho}}$ & $\sqrt{1 + \frac{2(1+\sqrt{2})}{n\rho}}$                              \\ \bottomrule
\end{tabular}}
\label{tab:TWR}
\end{table}
We can obtain a lower bound for TWR by dropping the factor $\frac{N-1}{N-n}$ and minimizing over $p$.
We can see that the width ratio mainly depends on $p$ and the relative magnitude between $n$ and $\rho$.
If $p$ is extreme (tends to 0 or 1), TWR is drastically large; when $p$ is around 0.5, TWR is close to the lower bound. Also, the added noise induces a term involving $\rho$.
For example, under the regime $\rho = 1/n$, the three algorithms result in an interval of length at least $\sqrt{3} \approx 1.73$, $\sqrt{5} \approx 2.24$, and $\sqrt{3+2\sqrt{2}} \approx 2.41$ as wide, respectively. 
It is trivial that with one stratum, {StrNz-PubSz} produces a tighter confidence interval than {PopNz-PubSz} does in that the ratio of $V_{\text{ex}}$ in (\ref{eq:weight1}) 
is 1/2. However, {PopNz-PubSz} will outperform {StrNz-PubSz} once there are enough strata such that (\ref{eq:weight1}) is greater than 1. 

\section{Numerical Results} 
In this section, we conduct both simulation studies and applications to assess and illustrate the numerical performance of the proposed algorithms. The budgeting $\rho_1 = \rho_2 = \rho/2$ are used for {PopNz-PubSz} and {StrNz-PrivSz}. We clip the proportions $\tilde p_h$ onto $[0,1]$ as mentioned in Remark \ref{rmk:clipping}. 

\subsection{Simulations}

We set up a set of experiments to (1) check the normality of noisy estimators, and (2) evaluate the performance of the proposed confidence intervals by varying the number of strata $H$, the true population proportion $p$, and the privacy level $\rho$. 
To generate the data, we need to specify the strata sizes $N_h$ and the sampling rates $r_h$. The setup of these parameters is presented in Table \ref{tab:parameters}. 
We generate a proportion for each stratum to create heterogeneity across strata. The true population proportion is then calculated and reported in each experiment.
\begin{table}[H]
\caption{Parameter setup. The resulting sample sizes are between 60 and 160.}
\resizebox{\columnwidth}{!}{%
\begin{tabular}{@{}c|cl||cc|c@{}}
\toprule
Fixed parameter & Value / Distribution &  &  & Varying parameter & Value / Distribution \\ \cmidrule(r){1-2} \cmidrule(l){5-6} 
$\alpha$ & 0.1 &  &  & $H$ & 1 or 20 \\
$N_h$ & Uniform(1500, 2000) &  &  & $p_h$ & 0.5, Uniform(0.4, 0.6) or Uniform(0.05, 0.15)\\
$r_h$ & Uniform(0.04, 0.08) &  &  & $\rho$ & $1/\max(n_h)$ or specified in the axis of the plot \\ \bottomrule
\end{tabular}}

\label{tab:parameters}
\end{table}

\subsubsection{Normality Check}
We first check whether the distributions of $\tilde p$ in the three algorithms are reasonably close to the theoretical normal distributions with the corresponding means and variances.  
Figure \ref{fig: qq} displays the Q-Q plots of the theoretical distribution  of $\tilde p$ versus its sample distribution:
\begin{itemize}
    \item Non-private: $\displaystyle\mathcal{N}(p, \operatorname{Var}(\hat p))$;
    \item {StrNz-PubSz}: $\displaystyle\mathcal{N}(p, \operatorname{Var}(\tilde p))$ as $\operatorname{Var}(\tilde p)$ in (\ref{eq:true var stra});
    \item {PopNz-PubSz}: $\displaystyle\mathcal{N}(p, \operatorname{Var}(\tilde p))$ as $\operatorname{Var}(\tilde p)$ in (\ref{eq:true var pop});
    \item {StrNz-PrivSz}: $\mathcal{N}\left( p + \sum_{h = 1}^H  \frac{w_h p_h}{2 \rho_2 n^2_h}, \sum_{h=1}^H w_h^2 V_h \right)$ with $V_h$ specified in (\ref{eq:true varh priv}).
\end{itemize}
 Note that, $\tilde p$ in Algorithms {StrNz-PubSz} and {PopNz-PubSz} are unbiased for $p$ while $\tilde p$ in {StrNz-PrivSz} is not. 
Nevertheless, under the condition that $\rho_2 = \omega(1/n_h)$ in Theorem \ref{alg:{StrNz-PrivSz}}, the bias term $ \sum_{h = 1}^H  \frac{w_h p_h}{2 \rho_2 n^2_h}$ is negligible and thus we do not make a bias correction in Algorithm \ref{alg:{StrNz-PrivSz}}. We observe great alignments between the theoretical and experimental distributions, indicating that the private estimators in all three algorithms are indeed highly close to being normally distributed.
\begin{figure}
\includegraphics[scale = 0.40
    ]{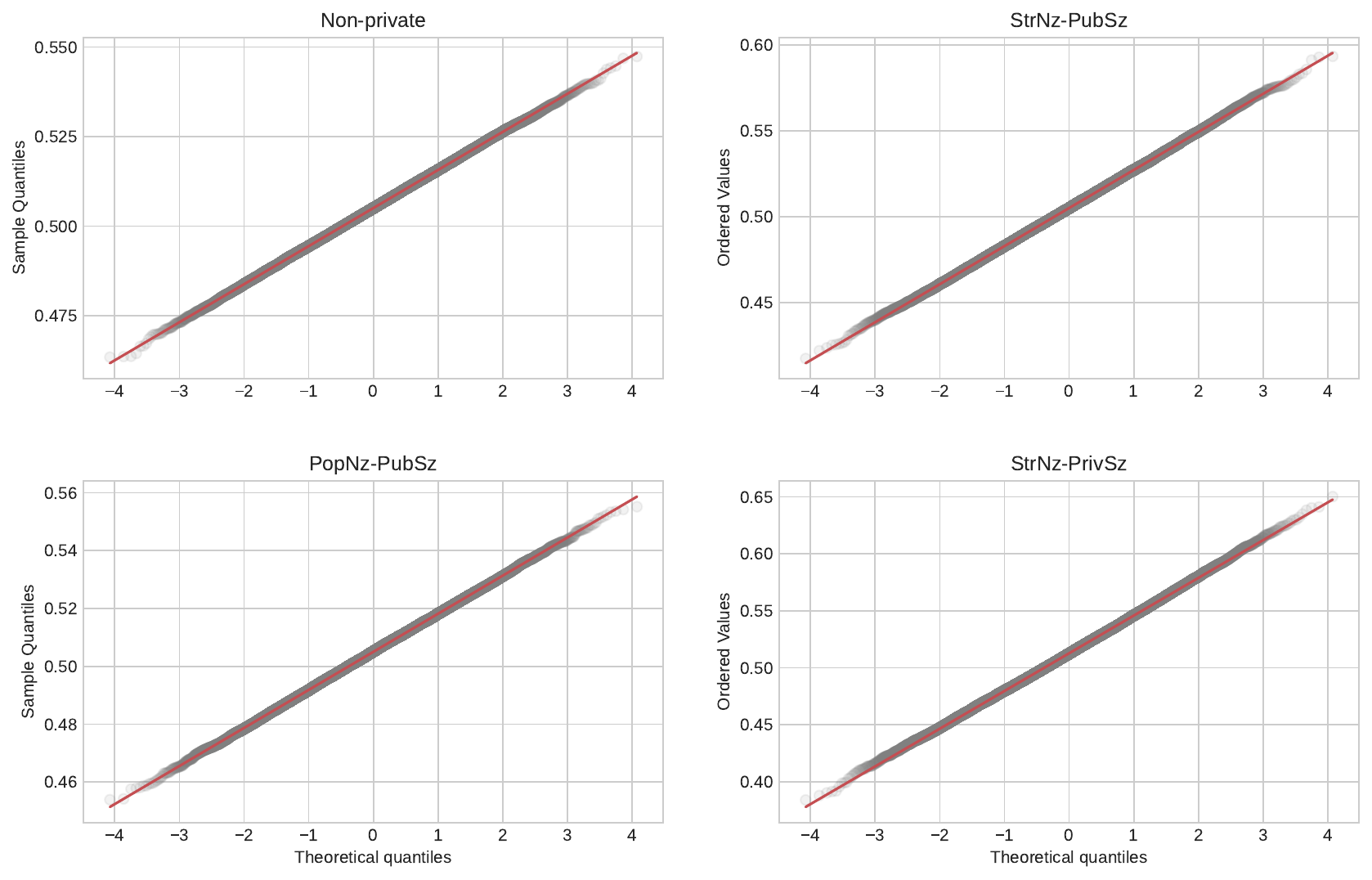} 
\caption{Q-Q plots: Theoretical versus sample distributions of $\tilde p$
    with 20 strata and $p = 0.505$ (resulting from $p_h \sim \textup{Uniform}(0.4, 0.6)$), based on 10,000 repetitions each.}
\label{fig: qq}
\end{figure}

\subsubsection{Varying Key Parameters}
Assured by the results of the normality check, we experiment with a wide range of the privacy budget, different numbers of strata, and true population proportions.

We examine the impact of $\rho$ on the performance of the three private estimators. 
The simulation is run on 10,000 repetitions and therefore the empirical coverage falling into $90\% \pm 0.006$ (departure of two standard deviations) is considered appropriate.
In Figure \ref{fig:coverage}, the empirical coverage is reasonable except that {StrNz-PrivSz} gives unnecessarily higher coverage when $\rho$ is smaller than around 0.005. 
This is because the budget is so small for the method that, with clipping, it covers the truth more often than needed. In this case, the confidence intervals are too wide to be as useful, as shown in Figure \ref{fig:width}.
For all three methods, the width grows as $\rho$ becomes smaller. However, the rates of width growth differ: in the multiple strata case we simulate, the width of {PopNz-PubSz} grows the slowest,
{StrNz-PrivSz} grows the fastest, and {StrNz-PubSz} is
in the middle.
Thus, the optimal privacy level should be chosen by taking into account the method, width, and coverage.  
For instance, if we want a 90\% of confidence level and width under 0.1, one can choose the value for $\rho$ as small as (1) $0.001$ for {PopNz-PubSz}, (2) $0.003$ for {StrNz-PubSz}, and (3) $0.01$ for {StrNz-PrivSz}.

\begin{figure}[ht]
     \centering
     \begin{subfigure}[b]{2.2in}
         \centering
         \includegraphics[width=2.2in]{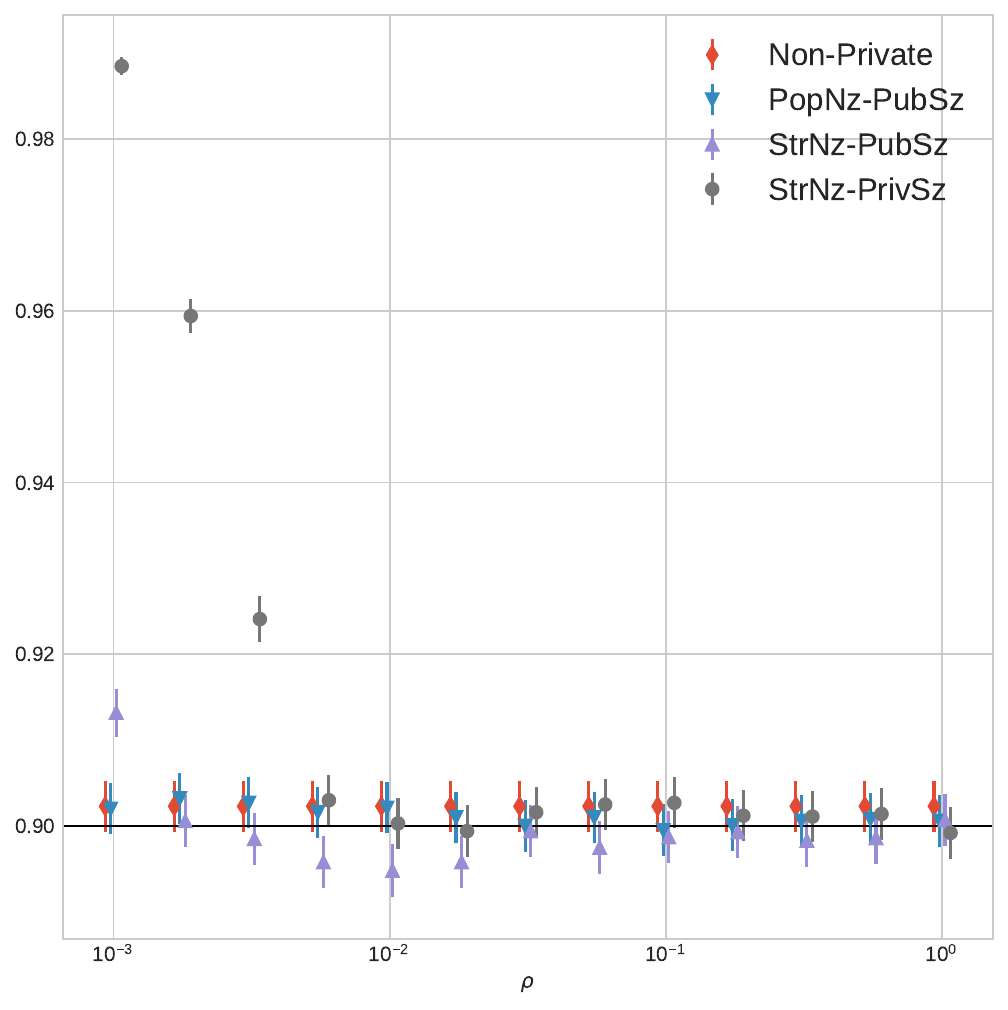}
         \caption{Empirical coverage}
         \label{fig:coverage}
     \end{subfigure}
     \hfill
     \begin{subfigure}[b]{2.4in}
         \centering
         \includegraphics[width=2.4in]{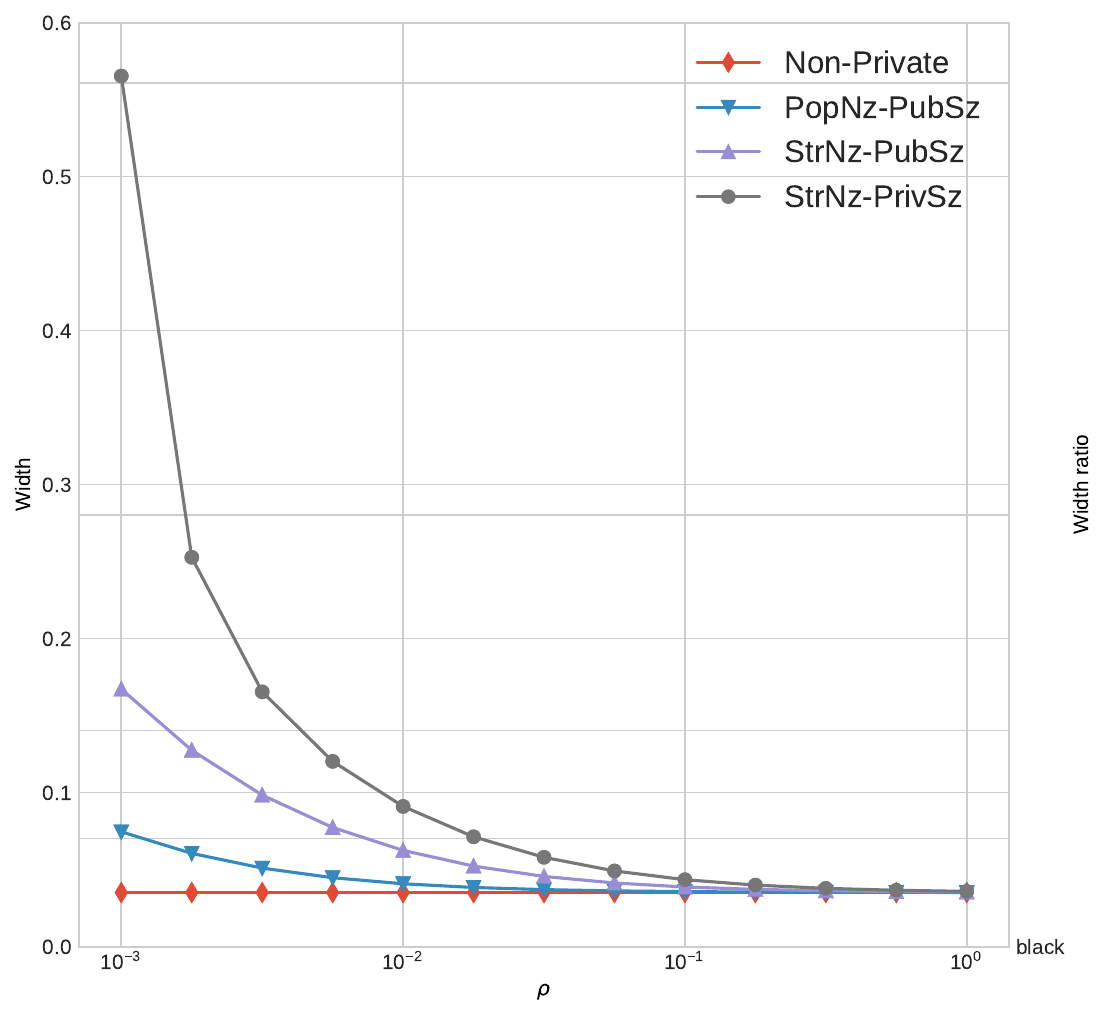}
         \caption{Width and ratio
         }
         \label{fig:width}
     \end{subfigure}

     \caption{Setup: 20 strata and $p = 0.505$ ($p_h$ $\sim$ \textup{Uniform}(0.4, 0.6)) with 10,000 repetitions. Figure (a) 
     is the empirical coverage with the black solid line indicating the nominal confidence level of 90\%. Error bars of one standard deviation are shown for coverage. The average width and width ratio are displayed in (b) with the non-private as the benchmark. Error bars of width are not visible in the plots and therefore not shown. }
    \label{fig:sim}
\end{figure}

In addition, Table \ref{tab:sim_results} 
shows the numerical results of three experiments with different combinations of the numbers of strata and the true population proportions.
The simulation in the middle panel shares the same setting as the experiment shown in Figure \ref{fig:sim} but has a fixed privacy level: $1/\max(n_h)$. This is an analogous regime to $\rho = 1/n$ (for simple random sampling) for multiple strata.
In the literature on differential privacy, the regime $\rho = 1/n$ for a simple random sample is often considered to understand how small $\rho$ can be as the sample size increases. Recall that a smaller $\rho$ means a higher privacy level.

As argued above, clipping $\tilde p_h$ (or $\tilde p$) onto [0,1] will yield better results in some cases. 
The conclusions coincide with the analyses in Section 4.
The empirical coverage of the three private ones in all simulations achieves the nominal level of 90\%, as guaranteed by Theorems \ref{thm:NoiStraPubSize}, \ref{thm:NoiPopPubSize}, and \ref{thm:NoiStraPrivSize}. The case where {StrNz-PrivSz} gives a 91.9\% confidence level in the bottom panel is due to clipping. (When the stratum proportions are close to the extreme, clipping 
is more noticeable.)

\begin{table}[ht]
\caption{Simulation results under $\rho = 1/n$ (or $\rho = 1/\max(n_h)$) regime based on 10,000 repetitions. The strata sizes and sampling rates are drawn as described in Table \ref{tab:parameters}. For the multiple strata case, the resulting sample sizes in  $n_h$ range from 72 to 152, and $\rho$ is set to be $1 / 152  \approx 6.58 \times 10^{-3}$. For the one-stratum case, we set the sample size to 152 so that we have the same level of privacy.}
\setlength{\tabcolsep}{13pt}
\renewcommand{\arraystretch}{1.25}
\resizebox{\columnwidth}{!}{%
\begin{tabular}{@{}c|cccc@{}}
\toprule
 & Non-Private & {StrNz-PubSz} & {PopNz-PubSz} & {StrNz-PrivSz} \\ \midrule
 & \multicolumn{4}{c}{$1$ stratum,  $p = 0.5$} \\
coverage & 0.893 & 0.901 & 0.894 & 0.901 \\
coverage SD & 3.09$\times 10^{-3}$ & 2.99$\times 10^{-3}$ & 3.08$\times 10^{-3}$ & 2.99$\times 10^{-3}$ \\
width & 0.127 & 0.228 & 0.295 & 0.327 \\
width SD & $5.47 \times 10^{-4} $ & $9.85 \times 10^{-4}$  & $ 8.89 \times10^{-3}$ & $ 3.15 \times 10^{-2}$ \\
CI & (0.436, 0.564) & (0.386, 0.614) & (0.352, 0.648) & (0.34, 0.667) \\
WR & 1 & 1.786 & 2.318 & 2.567 \\ \cmidrule(l){1-5} 
\textbf{} 
& \multicolumn{4}{c}{$20$ strata,  $p = 0.505$ ($p_h \sim \text{Uniform}(0.4, 0.6)$)} \\
coverage & 0.902 & 0.895 & 0.902 & 0.902 \\
coverage SD & 2.97$\times 10^{-3}$ & 3.07$\times 10^{-3}$ & 2.97$\times 10^{-3}$ & 2.97$\times 10^{-3}$ \\
width & 0.035 & 0.073 & 0.043 & 0.111 \\
width SD & $1.08 \times 10^{-4} $ & $1.58 \times 10^{-4}$  & $ 5.87 \times10^{-4}$ & $ 5.22 \times 10^{-3}$ \\
CI & (0.488, 0.523) & (0.469, 0.542) & (0.483, 0.527) & (0.457, 0.568) \\
WR & 1 & 2.074 & 1.239 & 3.168 \\
\cmidrule(l){1-5} 
 & \multicolumn{4}{c}{$20$ strata,  $p = 0.103$ ($p_h \sim \text{Uniform}(0.05, 0.15)$)} \\
coverage & 0.902 & 0.919 & 0.904 & 0.899 \\
coverage SD & 2.97$\times 10^{-3}$ & 2.73$\times 10^{-3}$ & 2.95$\times 10^{-3}$ & 3.01$\times 10^{-3}$ \\
width & 0.021 & 0.067 & 0.033 & 0.096 \\
width SD & $6.17 \times 10^{-4} $ & $5.21 \times 10^{-4}$  & $ 8.71 \times10^{-4}$ & $ 3.94 \times 10^{-3}$ \\
CI & (0.092, 0.113) & (0.073, 0.143) & (0.086, 0.119) & (0.072, 0.168) \\
WR & 1 & 3.189 & 1.571 & 4.563 \\
 \bottomrule
\end{tabular}}

\label{tab:sim_results}
\end{table}

The average width and width ratio (WR) varies. 
With one single stratum, WRs are near the lower bounds of theoretical width ratios (TWR) given in Section 4.2.2, which suggests that the lower bounds are almost tight. {StrNz-PubSz} gives a narrower CI than {PopNz-PubSz} with one stratum. 
But with more strata, {PopNz-PubSz} outperforms {StrNz-PubSz} in terms of WR. Having more strata means splitting the total privacy budget into smaller portions, which leads to adding more noise on the whole. The CI needs to be wider to achieve the same confidence level.
As for {StrNz-PrivSz}, however, it always yields the widest CI due to the additional price it pays to protect sample sizes simultaneously. 
On the other hand, with the same number of strata (20 here), we see that more extreme $p_h$ leads to a larger WR than $p_h$ in the middle range. This is because the factor $p_h(1-p_h)$ comes into play as $p(1-p)$ does in TWR in Table \ref{tab:TWR} for the one stratum case.

We also provide the sample standard deviation of the widths (width SD). In general, the non-private method results in a smaller standard deviation than the private ones. In some cases, clipping helps reduce 
the width SD for the private algorithms. With the same privacy level, there is more fluctuation in width for {PopNz-PubSz} compared to {StrNz-PubSz}. This is because we use one-half of the privacy budget and directly add noise to the variance estimate.
As expected, {StrNz-PrivSz} has the largest width SD since the magnitude of width is the largest and the ratio variable is heavy-tailed by design.
Nevertheless, compared to the width, the width SD for all methods is so small that it does not compromise the effectiveness of the confidence interval.

\subsection{Applications}
In this section, we apply the proposed methods to the 1940 Census full enumeration from IPUMS USA \cite{1940census} and evaluate the performance of three differentially private confidence intervals. 
To conduct stratified random sampling on the data set, the state-level geographical variable ``STATEICP'' (49 categories, constituting the then-48 states and Washington, D.C.) is used for stratification. 
Under stratified random sampling with $H=49$ strata,
we estimate the national unemployment rate for the first application. In the second application, we are interested in studying the discrepancy in income levels between black and white men.

\subsubsection{Confidence Intervals for the Unemployment Rate}
As an important indicator of the status of the national economy, the unemployment rate is the percentage of unemployed workers in the total labor force consisting of both the employed and unemployed. 
Thus, we consider all the individuals who are either employed or unemployed as the whole population.
In the 1940 Census data set, the binary characteristic ``EMPSTAT'' represents employment status. The full enumeration is considered the truth and the true population proportion is $p = 9.346\%$. 
To carry out stratified random sampling, sample sizes or sampling rates are selected for all 49 strata. For modern relevance, we simulate in a manner intended to mimic the canonical design implemented in the current American Community Survey (ACS), by choosing a typical range of sampling rates used in ACS which is $[0.5\%, 15\%]$. See Table \ref{tb:sampling_rates1} for detail.

\begin{table}[H]
\caption{Sampling rates}
\label{tb:sampling_rates1}
\begin{tabular}{@{}l|l@{}}
\toprule
Stratum size                          & Sampling rate \\ \midrule
$n_h \leq 5 \times 10^4$                    & $15\%$        \\
$5\times 10^4 < n_h  \leq 10^5$      & $10\%$        \\
$10^5 < n_h  \leq 5\times 10^5$     & $5\%$         \\
$5\times 10^5 < n_h  \leq  10^6$  & $2\%$         \\
$10^6< n_h  \leq 5\times 10^6$ & $1\%$         \\
$n_h >5\times 10^6$                    & $0.5\%$       \\ \bottomrule
\end{tabular}

\end{table}

To apply and assess the proposed algorithms, we experiment with a wide range of small privacy budgets: $\rho \in [10^{-6}, 10^{-3}]$. Each method is repeated 10,000 times and the empirical coverage, the average CI width, and the average CI width ratio (WR)
are computed.
As shown in Figure \ref{fig:app1_cov}, the empirical coverage is always around the nominal level which is chosen at the level of 90\% for the whole range of privacy levels. In Figure \ref{fig:app1_width}, the CI width and CI width ratio with the non-private CI as the benchmark, share the same shape. Even when the CI given by {StrNz-PrivSz} is 8 times the non-private CI width, the CI width is only 0.01 due to the large sample size. Both CI width and width ratio should be taken into account when choosing an optimal privacy level.

\begin{figure}[ht]
     \centering
     \begin{subfigure}[b]{2.2in}
         \centering
         \includegraphics[width=2.2in]{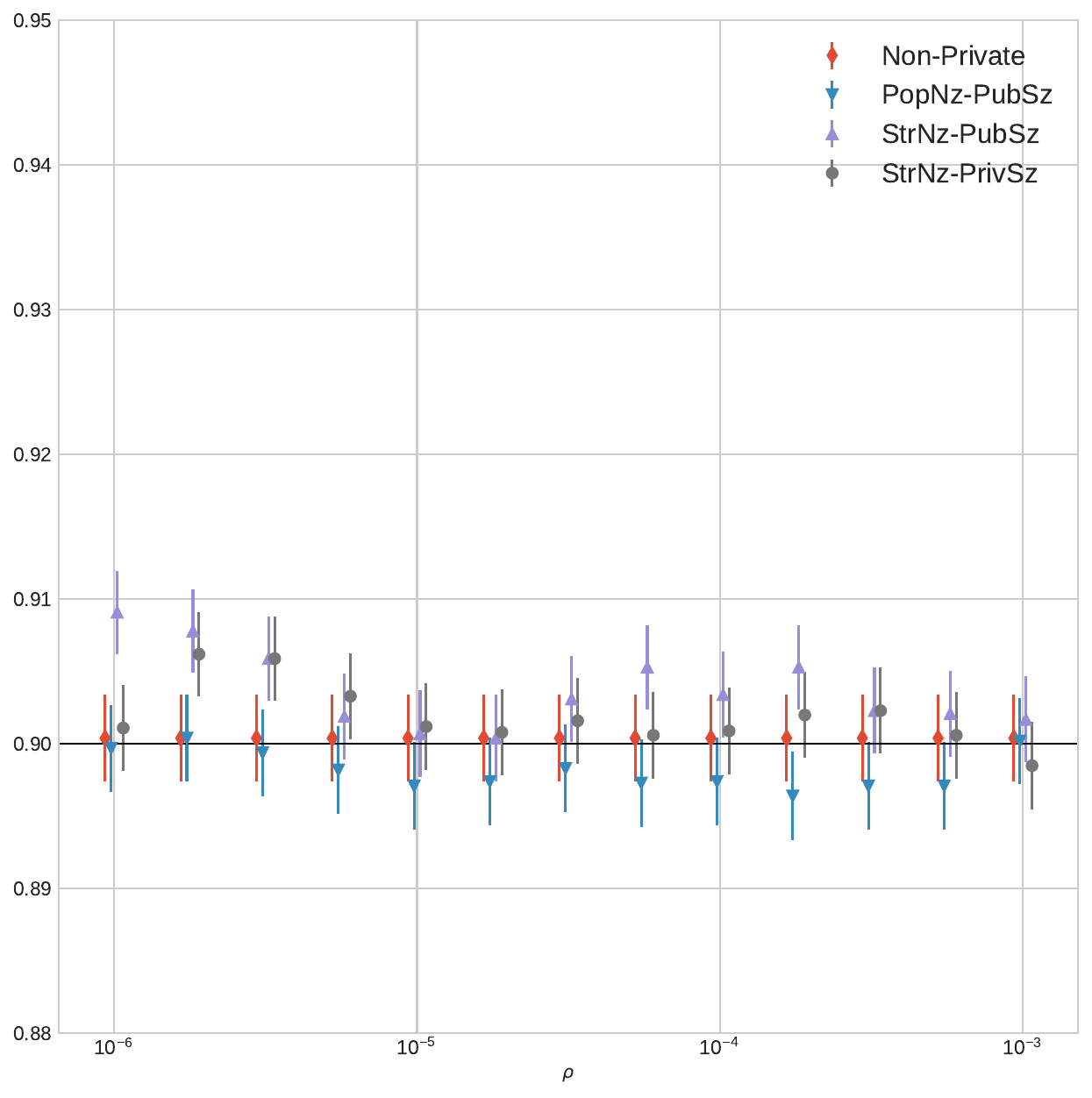}
         \caption{Empirical coverage}
         \label{fig:app1_cov}
     \end{subfigure}
     \hfill
     \begin{subfigure}[b]{2.4in}
         \centering
         \includegraphics[width=2.4in]{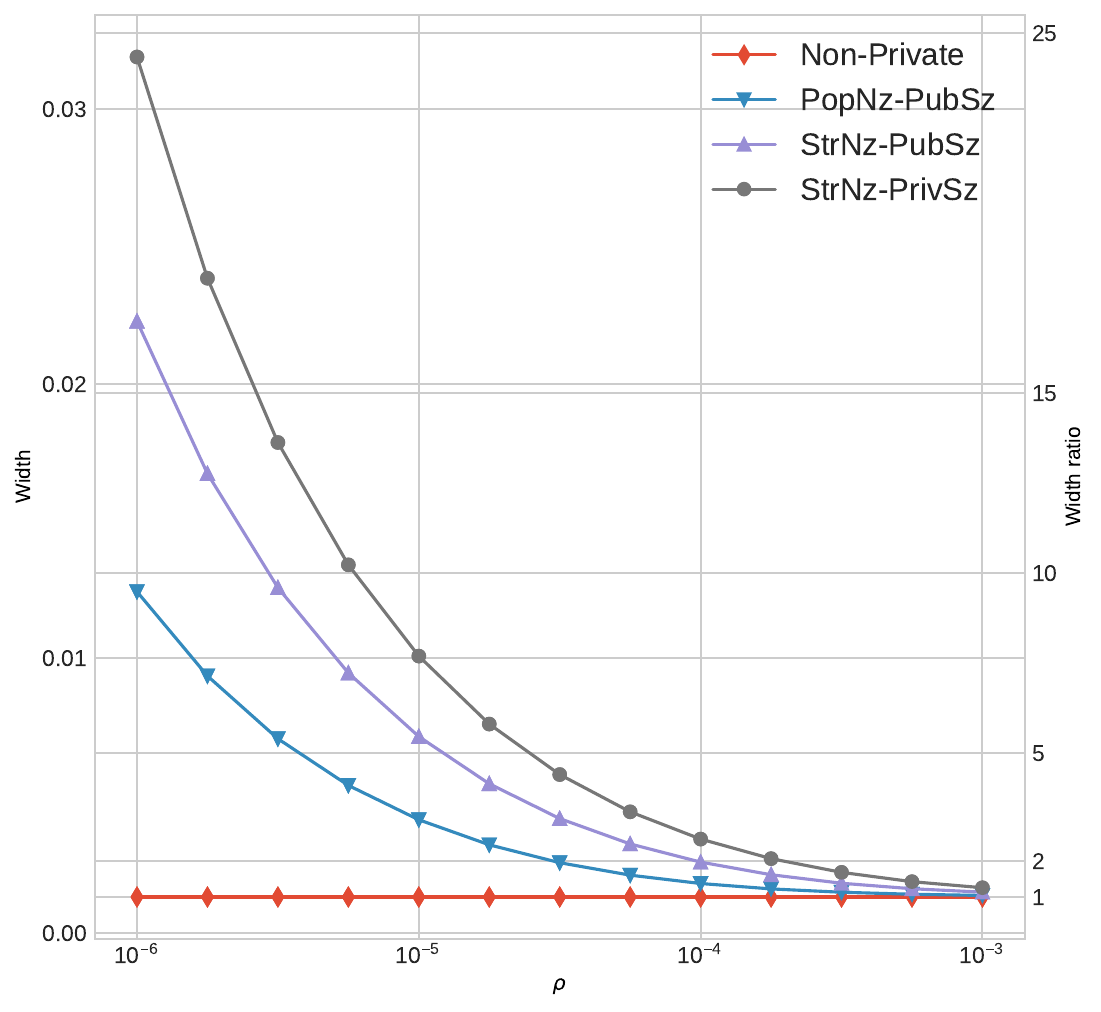}
         \caption{Width and ratio}
         \label{fig:app1_width}
     \end{subfigure}
     \caption{The empirical coverage with error bars, average width and width ratio of DP-CIs of the unemployment rate.}
    \label{fig:app1}
\end{figure}

\subsubsection{Confidence Intervals for the Difference in Income Level}
In the second application, we want to investigate whether there was a discrepancy between the income levels of white males (population 1) and that of black males (population 2). Note that only those who had valid income numbers in the 1940 Census are considered. Since the poverty thresholds were not developed until the 1960s and thus are not available for the 1940 data, the national income average is used as a threshold instead. We are interested in examining the difference in subpopulation proportions of those whose income levels passed this threshold. 

The geographic feature ``STATEICP'' is used for stratification, yielding 49 strata, with stratum size ranges of $ (41838, 4621442)$ for the population of white males and $(50, 309214)$ for the population of black males. Sampling rates are adaptively chosen based on stratum sizes. For the population of white males, the range of sampling rates is also $[0.5\%, 15\%]$, whereas the range of sampling rates is $[0.5\%, 30\%]$ for the population of black males given its small stratum sizes.
Additionally, to allow solid approximations based on the asymptotic results, we impose that the sample sizes are adjusted to be 50 if the sampling rates give smaller sizes than 50. See Table \ref{tab:rates_app2}
 for detail.

\begin{table}[!htb]
\caption{Sampling rates for two populations. Stratum sizes $n_h \in (4.1\times 10^4, 4.7\times 10^6)$ for the population of white males and stratum sizes $n_h \in (50, 3.1 \times 10^5)$ for the population of black males. *The sample size will be adjusted to be 50 if the above sampling rate results in a size smaller than 50.} 
\resizebox{1.1\columnwidth}{!}{%
    \begin{minipage}[c]{.5\linewidth}
      \centering 
        \begin{tabular}{@{}l|l@{}}
\toprule
Stratum size $N_h $ of white males                     & Sampling rate \\ \midrule
$N_h \leq 5 \times 10^4$                    & $15\%$        \\
$ 5 \times 10^4 < N_h \leq  10^5$      & $10\%$        \\
$ 10^5 < N_h\leq 5 \times 10^5$     & $5\%$         \\
$5 \times 10^5 < N_h \leq  10^6$  & $2\%$         \\
$ 10^6 < N_h \leq 4 \times 10^6$ & $1\%$         \\
$N_h> 4 \times 10^6$                    & $0.5\%$       \\ \bottomrule
\end{tabular}
    \end{minipage}%
    \begin{minipage}[c]{0.8\linewidth}
      \centering
        \begin{tabular}{@{}l|l@{}}
\toprule
Stratum size $N_h $ of black males                & Sampling rate* \\ \midrule
$N_h \leq 500$                & $30\%$        \\
$500 < N_h \leq 5 \times 10^3$        & $15\%$        \\
$5 \times 10^3< N_h \leq 10^4$   & $5\%$         \\
$10^4 < N_h \leq 2 \times 10^4$  & $2\%$         \\
$2 \times 10^4 < N_h \leq 3 \times 10^4$ & $1\%$         \\
$n_h > 3 \times 10^4$                & $0.5\%$       \\ \bottomrule
\end{tabular}
    \end{minipage} 
    }
\label{tab:rates_app2}
\end{table}

Let $p_1$ and $p_2$ denote the proportions of eligible individuals who earned more than the national average income level \$442.12.
The true values of proportions are $p_1 = 49.0223\%$ and $p_2 = 29.5152\%$.
Let $ p_{\text{diff}} =  p_1 -  p_2$, then the true difference in these two proportions is  $ p_{\text{diff}} = 19.5071\%$.
By the additivity of two independent normal distributions, naturally, we use the following differentially private CI:
\begin{equation}
    \tilde p_{\text{diff}} + z_{1-\alpha/2} \sqrt{\widetilde{\operatorname{V}}(\tilde p_{\text{diff}})},
\end{equation}
where $\widetilde{V}(\cdot)$ denotes a private estimator of variance, $\tilde p_{\text{diff}}$ is defined as $\tilde p_1 - \tilde p_2$ and $$ \widetilde{\operatorname{V}}(p_{\text{diff}}) =
\widetilde{\operatorname{V}}(p_1) + \widetilde{\operatorname{V}}(p_2)
.$$

\begin{figure}[ht]
     \centering
     \begin{subfigure}[b]{2.2in}
         \centering
         \includegraphics[width=2.2in]{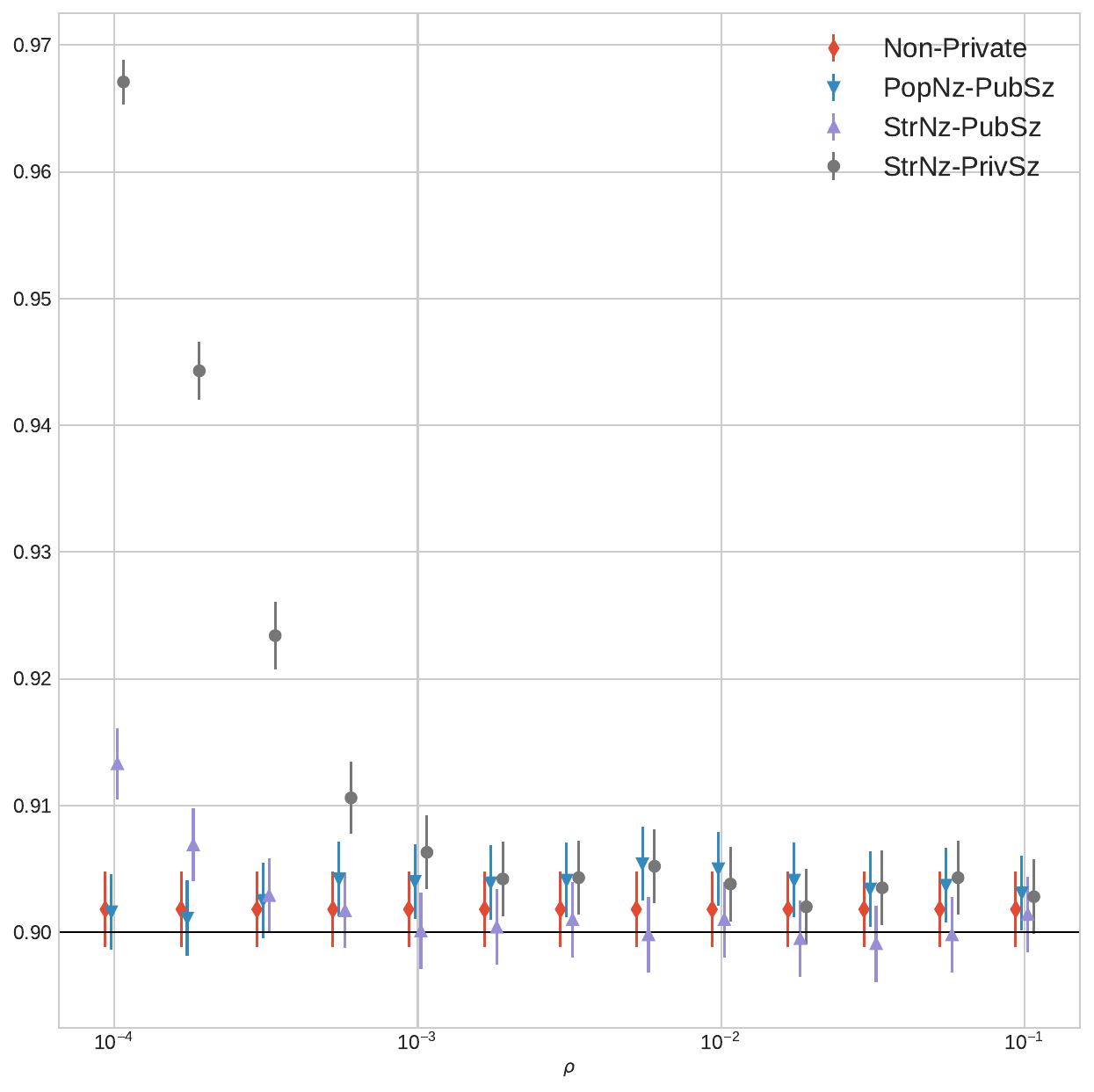}
         \caption{Empirical coverage}
         \label{fig:app2_cov}
     \end{subfigure}
     \hfill
     \begin{subfigure}[b]{2.4in}
         \centering
         \includegraphics[width=2.4in]{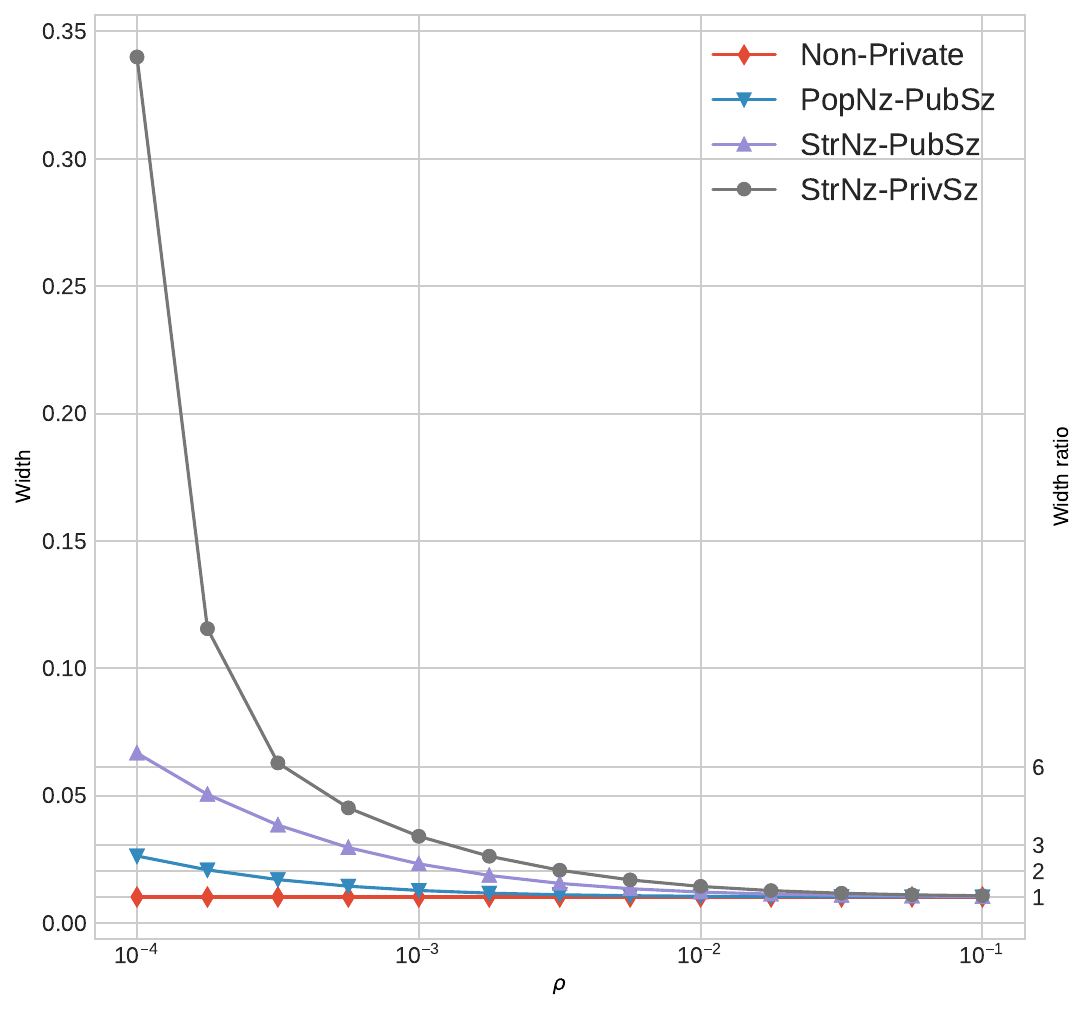}
         \caption{Width and ratio
         }
         \label{fig:app2_width}
     \end{subfigure}
     \caption{The empirical coverage with error bars, average width and width ratio of DP-CIs of the difference of the above-national-income-level proportions between black and white males with valid income values.}
    \label{fig:app2}
\end{figure}

In Figure \ref{fig:app2}, similar patterns are observed in this application as in the first. All CIs have empirical coverage around/above the nominal confidence level as in the simulation study in Section 5.1.2. The phenomenon of higher coverage is due to small $\rho$ and effective clipping. 
When the range of stratum sizes is large (it is (50, 309214) in this application), that is, when the stratum sizes are very different, 
a large privacy budget $\rho$ should be chosen. The choice of a small $\rho$ harms the estimates of small-sized strata. We advise that the smallest $\rho$ be chosen given the tolerance of uncertainty in terms of width and/or width ratio. 
For example, if the accuracy requirement is that the width should be under 0.05 or WR under 5, then the  best choices of $\rho$ among the experiments in Figure \ref{fig:app2_width} are (1) $0.0001$ for {PopNz-PubSz}, (2) 0.0018 for {StrNz-PubSz}, and (3) 0.0056 for {StrNz-PrivSz}.

\section{Discussion}
We have designed three algorithms to construct confidence intervals for the population proportion under stratified random sampling with zero concentrated differential privacy guarantees. We consider both the case where the sample sizes are public and the case where they are private information.
Theoretical results including privacy guarantees and asymptotic properties are established. With proper conditions on the relation between the privacy budget and sample sizes, as stated in the theorems, the resulting confidence intervals will achieve the desired coverage asymptotically, and the width tends to be that of a non-private confidence interval when the sample sizes go to infinity.

In the simulation studies and two applications, we have experimented with a wide range of privacy budgets under a variety of parameter setups. The three algorithms always perform well in terms of empirical coverage. The width and width ratio are in a reasonable range even under the strict regime where $\rho = 1/\max_h{n_h}$.
Typically in practice, a constant between 0.001 to 10 is chosen to be the privacy budget. According to our experiments, with the choice of the smallest budget in this range, 0.001, the three algorithms still have fairly good results even when the smallest stratum has only a size 50 (as demonstrated in the second application). 

The comparative analysis of the three algorithms in Section 4.2 gives actionable guidance to practitioners. When releasing the population proportion is the only goal and there are enough strata (such that  Eq.(\ref{eq:weight1}) regarding sample weights is greater than 1), {PopNz-PubSz} is the better option. However, if stratum proportions should also be released or there are just a few strata, {StrNz-PubSz} is preferable.
On the other hand, when the population proportion and sample sizes must be protected simultaneously, {StrNz-PrivSz} is the only algorithm presented in this paper. {StrNz-PrivSz}, compared to the case with public sample sizes, needs a larger budget to meet the same width requirement on account of the additional cost of protecting sample sizes.

There are a few open questions worth considering for future research. 
In this paper, we discuss the classic case where the number of strata is fixed, and the sample sizes tend to infinity. In principle, asymptotic normality is also valid in other settings with finite sample sizes. For example, it has been shown in the non-private setting that as the total sample size $N\to\infty$ with many small samples or a few large samples, or some combination thereof, central limit theorems hold under certain (complex) conditions \citep{Bickel1984}.
Under the constraints of differential privacy, we have shown that the trade-off between the privacy and accuracy (a.k.a., utility) of DP-CIs depends on the smallest sample, i.e, recall the condition, $\rho = \omega(1/n_h)$ for all $h$, in Section \ref{sec:thmresults}.
The overall privacy loss is determined by the largest privacy loss among all strata.
However, when the strata sizes $N_h$ remain finite, the weights ($w_h = N_h/N$) of these strata tend to 0 as $N\to\infty$. Therefore, the noise injected into these small strata should not harm the overall accuracy of the intervals if $N$ is sufficiently large.

More interestingly, we do not provide `PopNz-PrivSz'  -- an analogous algorithm to {PopNz-PubSz} for the private sample sizes case. To protect both the population proportion and the sample sizes, the direct addition of noise to the non-private aggregated estimator is not plausible. One should consider more sophisticated mechanisms other than directly adding noise to the statistics. If `PopNz-PrivSz' were proposed, we shall expect it to yield a narrower confidence interval since we only need to publish the private population proportion without being able to provide private confidence intervals for stratum proportions at the same time. 

Another direction for future research would be optimal budget allocation. We do not discuss how to best divide the total budget for {PopNz-PubSz} or {StrNz-PrivSz}. 
Budgeting for the composed application of the algorithms may also be of interest, like in Section 5.2.2 where we apply the algorithms twice for two independent populations. 

Lastly, one broad direction is to develop the differentially private versions for other alternatives to the basic Wald interval, such as the Wilson Interval, Jeffreys interval, etc.(see~\cite{Franco2019ComparativeSO} for a comparative summary of seven such types of confidence intervals for proportions). Many of these latter are specifically designed for the case of small sample sizes, which we do not consider here and for which we expect fundamentally different approaches to differential privacy likely to be necessary.

\appendix

\section{Proofs}\label{app}
\renewcommand{\thesection}{\Alph{section}}

\subsection{Proof of Theorem \ref{prop:condmom}}
\begin{lemma}
Let $X \sim \mathcal{N}(\mu, \sigma^2)$ and $S = \{\mu-a\leq X \leq \mu + a\}$. For any $a>0$ and an integer $k\geq 1$, the conditional even moments
\begin{equation}
    \mathbb{E}[(X - \mu)^{2k}\mid S] =  \sigma^{2k}(2k-1)!!-O\left(e^{-\frac{a^2}{2\sigma^2}} a^{2k-1}\right),
\end{equation}
where the big-O hides a constant depending on $\sigma$ and $k$.
\label{lem: moments}
\end{lemma}
\begin{proof}
Without loss of generality, we assume $\mu = 0$. We prove the lemma by induction. 
Set $k = 1$, integrate by parts,
\begin{equation*}
\begin{aligned}
    \mathbb{E}[X^{2}I_S] & = \int_{-a}^{a}x^2 \frac{1}{\sigma \sqrt{2 \pi}} e^{-\frac{x^2}{2\sigma^2}}\,dx \\
    & = \frac{\sigma}{ \sqrt{2 \pi}} \left( -xe^{-\frac{x^2}{2\sigma^2}}  \big|_{-a}^{a} + \int_{-a}^{a} e^{-\frac{x^2}{2\sigma^2}}\,dx \right).
\end{aligned}
\end{equation*}
Integrate by substitution, the integral in the second term becomes
\begin{equation*}
    \int_{-a}^{a}e^{-\frac{x^2}{2\sigma^2}}\,dx
    =\sigma \sqrt{2\pi} \operatorname{erf}\left(\frac{a}{\sigma\sqrt{2}}\right)
\end{equation*}
where
\begin{equation*}
    \operatorname{erf}(z) = \sigma \int_{0}^{z}e^{-t^2}\,dt
\end{equation*}
is the \textit{error function}.
Then,
\begin{equation*}
    \mathbb{E}[X^{2}I_S] = \sigma^2 \operatorname{erf}\left(\frac{a}{\sigma\sqrt{2}}\right) - O\left(e^{-\frac{a^2}{2\sigma^2}}a\right).
    \end{equation*}
Assuming 
\begin{equation}
    \mathbb{E}[X^{2k}I_S] =  \sigma^{2k}(2k-1)!! \operatorname{erf}\left(\frac{a}{\sigma\sqrt{2}}\right) - O\left(e^{-\frac{a^2}{2\sigma^2}}a^{2k-1}\right),
    \label{eq:induction-k}
    \end{equation}
then integrate by parts for the $k+1$ case,
\begin{equation*}
\begin{aligned}
    \mathbb{E}[X^{2(k+1)}I_S] &= \int_{-a}^{a}x^{2k+2} \frac{1}{\sigma \sqrt{2 \pi}} e^{-\frac{x^2}{2\sigma^2}}\,dx \\
    & = \frac{\sigma}{ \sqrt{2 \pi}} \int_{-a}^{a} x^{2k+1} \cdot \frac{x}{\sigma^2} e^{-\frac{x^2}{2\sigma^2}}  \,dx \\
    & = \frac{\sigma}{ \sqrt{2 \pi}}\left( -x^{2k+1}e^{-\frac{x^2}{2\sigma^2}}  \big|_{-a}^{a} + (2k+1)\int_{-a}^{a}x^{2k}e^{-\frac{x^2}{2\sigma^2}}\,dx \right) \\
    & =  \sigma^2 (2k+1)\mathbb{E}[X^{2k}I_S] - O\left(e^{-\frac{a^2}{2\sigma^2}}a^{2k+1} \right).
\end{aligned}
\end{equation*}
Plug in (\ref{eq:induction-k}), we obtain
\begin{equation*}
    \mathbb{E}[X^{2(k+1)}I_S] = \sigma^{2k+2}(2k+1)!! \operatorname{erf}\left(\frac{a}{\sigma\sqrt{2}}\right) -O\left(e^{-\frac{a^2}{2\sigma^2}}a^{2k+1}\right).
\end{equation*}
So far we have proved (\ref{eq:induction-k}).
Note that
\begin{equation*}
    \Pr(S) = \int_{-a}^{a}\frac{1}{\sigma \sqrt{2 \pi}} e^{-\frac{x^2}{2\sigma^2}}\,dx = \operatorname{erf}\left(\frac{a}{\sigma\sqrt{2}}\right),
\end{equation*}
and that the image of $\operatorname{erf}(z)$ is between $(-1,1)$.
Therefore, 
\begin{equation*}
   \mathbb{E}[X^{2k}\mid S] = \mathbb{E}[X^{2k}I_S]/ \Pr(S) = \sigma^{2k}(2k-1)!!  - O\left(e^{-\frac{a^2}{2\sigma^2}}a^{2k-1}\right).
\end{equation*}
\end{proof}

\begin{proof}[Proof of (\ref{eq:condmean}) in Theorem \ref{prop:condmom}]
Consider the Taylor series of $\frac{1}{x}$ at $x = \mu$:
\begin{equation*}
    \frac{1}{x} = \sum_{j=0}^\infty \frac{(-(x-\mu))^j}{\mu^{j+1}}= \frac{1}{\mu} - \frac{x-\mu}{\mu^2}+ \frac{(x-\mu)^2}{\mu^3} - \frac{(x-\mu)^3}{\mu^4}+ \dots
\end{equation*}
Let $y_m$ be the partial sum of the above series, i.e., $y_m(x)=  \sum_{k=0}^m \frac{(-(x-\mu))^k}{\mu^{k+1}}$.
Then $y_m(x)$ converges to $\frac{1}{x}$ in $(0, 2\mu)$ which contains $[1, 2\mu - 1]$. Let 
\begin{equation*}
    g(x) = \sum_{k=0}^\infty \frac{|x-\mu|^k}{\mu^{k+1}} = 
    \begin{cases}
       \frac{1}{x}, & \text{if $1\leq x \leq \mu$}\\
       \frac{1}{2\mu - x} & \text{if $\mu < x \leq 2 \mu -1$}
    \end{cases}
\end{equation*}
Then $g$ is integrable as
\begin{equation*}
    \int_{1}^{2\mu - 1}|g(x)|d\nu = \int_{1}^{\mu} \frac{1}{x} d\nu + \int_{\mu }^{2\mu -1} \frac{1}{2\mu - x} d\nu
    = 2\int_{1}^{\mu} \frac{1}{x} d\nu <\infty,
\end{equation*}
where $d\nu = f(x)dx$ is induced by $\mathcal{N}(\mu, \sigma^2)$ conditional on event $S$.
Note also that $|y_m(x)| \leq g(x)$ for any naturals $m$ and $x \in [1, 2\mu-1]$. 
By the dominated convergence theorem, the operations of limit and integral are exchangeable for $y_m(x)$. 
\begin{equation}
    \begin{aligned}
      \int_{1}^{2\mu -1}  \frac{1}{x} d\nu & =  \int_{1}^{2\mu -1}  \lim_{m\rightarrow \infty} y_m(x) d\nu \\
     & =   \lim_{m\rightarrow \infty} \int_{1}^{2\mu -1}  y_m(x) d\nu\\
     & =  \lim_{m\rightarrow \infty} \int_{1}^{2\mu -1}  \left(\sum_{j=0}^m \frac{(-(x-\mu))^j}{\mu^{j+1}}\right) d\nu \\
     & = \lim_{m\rightarrow \infty} \left(\sum_{j=0}^m \int_{1}^{2\mu -1}  \frac{(-(x-\mu))^j}{\mu^{j+1}} d\nu  \right)\\
    \end{aligned}
\end{equation}
Then,
\begin{equation}
\begin{aligned}
       \mathbb{E}\left( \frac{1}{X} \mid S \right) & = \sum_{j = 0}^\infty \frac{1}{\mu^{j+1}}\mathbb{E}\left[(-(X - \mu ))^{j}  \mid S \right]\\
       & = \sum_{j = 0}^\infty \frac{1}{\mu^{2j+1}}\mathbb{E}\left[(X - \mu )^{2j}  \mid S \right]\\
       & = \sum_{j = 0}^k \frac{1}{\mu^{2j+1}}\mathbb{E}\left[(X - \mu )^{2j}  \mid S \right] + \frac{1}{\mu}\sum_{j = k+1}^\infty \mathbb{E}\left[\left(\frac{X-\mu}{\mu}\right)^{2j} \mid S \right].
\label{eq:cond_mean}
\end{aligned}
\end{equation}
The second equality is because the odd moments are zero due to symmetry.

Note that given event $S$, $|\frac{X-\mu}{\mu}| \leq \frac{\mu - 1}{\mu} < 1$, then
\begin{equation}
    \mathbb{E}\left[\left(\frac{X-\mu}{\mu}\right)^{2k+2} \mid S \right]\leq \left(\frac{\mu-1}{\mu}\right)^{2}\mathbb{E}\left[\left(\frac{X-\mu}{\mu}\right)^{2k} \mid S \right].
\label{ineq:factor}
\end{equation}
It follows that
\begin{equation*}
\begin{aligned}
     \sum_{j = k + 1 }^\infty \mathbb{E}\left[\left(\frac{X-\mu}{\mu}\right)^{2j}  \mid S \right] 
     &\leq \sum_{j = 0 }^\infty \left(\frac{\mu-1}{\mu}\right)^{2j} \mathbb{E}\left[\left(\frac{X-\mu}{\mu}\right)^{2k + 2} \mid S \right]  \\
     & = \frac{\mu^2}{2\mu -1}\mathbb{E}\left[\left(\frac{X-\mu}{\mu}\right)^{2k+2} \mid S \right]\\
     & = O\left( \frac{1}{\mu^{2k+1}}\right)\cdot 
     \mathbb{E}[\left({X-\mu}\right)^{2k+2} \mid S ].
\end{aligned}
\end{equation*}
Applying Lemma \ref{lem: moments}, by the choice of $a = \mu - 1$, (\ref{eq:cond_mean}) becomes
\begin{equation*}
\begin{aligned}
        \mathbb{E}\left( \frac{1}{X} \mid S \right) 
        & = \frac{1}{\mu}\sum_{j = 0}^{k} \frac{(2j-1)!!\sigma^{2j}}{\mu^{2j}}   + O\left(\frac{\sigma^{2k+2}}{\mu^{2k+2}}\right).
\end{aligned}
\label{eq:moments}
\end{equation*}
\end{proof}

\begin{proof}[Proof of (\ref{eq:condvar}) in Theorem \ref{prop:condmom}]
We conduct a similar procedure for the second moment of $X\mid S$. 
Based on the Taylor expansion
\begin{equation*}
    \frac{1}{x^2} = \sum_{j=0}^\infty \frac{(j+1)(-(x-\mu))^j}{\mu^{j+2}}= \frac{1}{\mu^2} - \frac{2(x-\mu)}{\mu^3}+ \frac{3(x-\mu)^2}{\mu^4} - \frac{4(x-\mu)^3}{\mu^5}+ \cdots,
\end{equation*}
we have 
\begin{equation}
\begin{aligned}
       \mathbb{E}\left( \frac{1}{X^2} \mid S \right) & = \sum_{j = 0}^\infty \frac{j+1}{\mu^{j+2}}\mathbb{E}\left[(-(X - \mu ))^{j}  \mid S \right]\\
       & = \sum_{j = 0}^\infty \frac{2j+1}{\mu^{2j+2}}\mathbb{E}\left[(X - \mu )^{2j}  \mid S \right]\\
       & = \sum_{j = 0}^k \frac{2j+1}{\mu^{2j+2}}\mathbb{E}\left[(X - \mu )^{2j}  \mid S \right] + \frac{1}{\mu^2}\sum_{j = k+1}^\infty (2j+1) \mathbb{E}\left[\left(\frac{X-\mu}{\mu}\right)^{2j} \mid S \right] .
\end{aligned}
\label{eq:2ndmom}
\end{equation}
Due to (\ref{ineq:factor}),
it follows that
\begin{equation}
    \begin{aligned}
     & \quad \sum_{j = k+1}^\infty (2j+1) \mathbb{E}\left[\left(\frac{X-\mu}{\mu}\right)^{2j} \mid S \right] \\
     &\leq \mathbb{E}\left[\left(\frac{X-\mu}{\mu}\right)^{2k + 2} \mid S \right]  \cdot \sum_{j = 0 }^\infty (2k+3 + 2j) \left(\frac{\mu-1}{\mu}\right)^{2j}  \\
     & = \mathbb{E}\left[\left(\frac{X-\mu}{\mu}\right)^{2k + 2} \mid S \right]\cdot \left[ (2k+3) \sum_{j = 0}^\infty \left(\frac{\mu-1}{\mu}\right)^{2j} + 2\sum_{j = 1 }^\infty j \left(\frac{\mu-1}{\mu}\right)^{2j} \right]\\
     & = \mathbb{E}\left[\left(\frac{X-\mu}{\mu}\right)^{2k+2} \mid S \right]\cdot \left[ \frac{(2k+3)\mu^2}{2\mu-1} + 2\frac{\mu^2(\mu-1)^2}{(2\mu-1)^2}\right]\\
     & = \mathbb{E}\left[\left(\frac{X-\mu}{\mu}\right)^{2k+2} \mid S \right]\cdot O(\mu^2)\\
     & = O\left( \frac{1}{\mu^{2k}}\right)\cdot 
     \mathbb{E}[\left({X-\mu}\right)^{2k+2} \mid S ],
    \end{aligned}
\end{equation}
where the term $\sum_{j = 1 }^\infty j \left(\frac{\mu-1}{\mu}\right)^{2j}$ is a sum of an arithmetic–geometric sequence. 
By Lemma \ref{lem: moments}, (\ref{eq:2ndmom}) becomes
\begin{equation}
\begin{aligned}
       \mathbb{E}\left( \frac{1}{X^2} \mid S \right)
        & = \frac{1}{\mu^2}\sum_{j = 0}^{k} \frac{(2j+1)!!\sigma^{2j}}{\mu^{2j}}   + O\left(\frac{\sigma^{2k+2}}{\mu^{2k+2}}\right).
\end{aligned}
\end{equation}
\end{proof}

\subsection{Proof of Theorem \ref{thm:privacy}}

\begin{proof}[Proof for Algorithm \ref{alg:strlevel}] Under neighboring relation $\sim_{ss}$, only one record changes within one stratum and sample sizes remain the same.
Applying the Gaussian mechanism to each stratum at the level of $\rho$ gives $\rho$-zCDP guarantee. By post-processing, the confidence interval is also $\rho$-zCDP.
\end{proof}

\begin{proof}[Proof for Algorithm \ref{alg:poplevel}]
The sensitivities of $\hat p$ and $\widehat{\operatorname{Var}}(\hat p)$ are $\Delta_p$ and $\Delta_V$, respectively.
Applying the Gaussian mechanism, it follows that $\tilde{p}$ is $\rho_1$-zCDP and $\widetilde{V}$ is $\rho_2$-zCDP. By basic composition, the confidence interval $\tilde{p} \pm z_{1-\frac{\alpha}{2}} \sqrt{\widetilde{V}}$ is $(\rho_1 + \rho_2)$-zCDP.
\end{proof}

\begin{proof}[Proof for Algorithm \ref{alg:{StrNz-PrivSz}}]
By the Gaussian mechanism and the basic composition property of zCDP, we know that $\tilde{p}_h$ is $\rho$-zCDP. Under neighboring relation $\sim_{r}$, only one record changes within one stratum. Then, by post-processing, the confidence interval is $\rho$-zCDP.
\end{proof}

\subsection{Proof of Theorem \ref{thm:NoiStraPubSize}}
Before proving the theorem, we revisit the finite-population CLT first:
\begin{theorem}[Theorem 1, \cite{XinranLi}]
Consider a finite population $\Pi = \{X_{1}, ..., X_{N}\}$ of size
$N$. Let $\mu$ be the population mean and $\bar X_n$ be the mean of a simple random sample of
size $n$ from $\Pi$, and ${\operatorname{Var}}(\bar X_n)$ is the variance of $\bar X_n$. The finite population variance of $\Pi$ is denoted by $$v = \frac{1}{N-1}\sum_{i=1}^N (X_{i} - \mu)^2.$$ 
As $N \rightarrow \infty$, if 
\begin{equation}
    \frac{1}{\min(n, N-n)}\cdot \frac{\max_{1\leq i \leq N}(X_{i} - \mu)^2}{v} \rightarrow 0,
\label{clt:cond}
\end{equation}
we have
\begin{equation}
    \frac{\bar X_n - \mu}{\sqrt{\strut \operatorname{Var}(\bar X_n)}} \stackrel{d}{\rightarrow} \mathcal{N}(0,1).
\end{equation}
\label{thm: SRS_CLT}
\end{theorem}

The variance of $\bar X_n$ is determined by the population variance $v$ which is unknown. Nevertheless, the sample variance $\widehat{\operatorname{Var}}(\bar X_n)$ can be used to estimate $v$. To make sure the CLT still holds when substituting ${\operatorname{Var}}(\bar X_n)$ by $\widehat{\operatorname{Var}}(\bar X_n)$, 
the consistency of $\widehat{\operatorname{Var}}(\bar X_n)$ is crucial, as stated in the following lemma.

\begin{lemma}
Let $\widehat{\operatorname{\operatorname{Var}}}(\bar X_n)$ be the sample variance. $\widehat{\operatorname{\operatorname{Var}}}(\bar X_n)$ is an unbiased estimator for ${\operatorname{\operatorname{Var}}}(\bar X_n)$.
Moreover, under the condition in Theorem \ref{thm: SRS_CLT},  as $N \rightarrow \infty$, $$\widehat{\operatorname{\operatorname{Var}}}(\bar X_n)/{\operatorname{\operatorname{Var}}}(\bar X_n) \stackrel{p}{\rightarrow} 1.$$
\label{prop: consistency}
\end{lemma}
Now we prove Theorem \ref{thm:NoiStraPubSize}:
\begin{proof}[Proof of Theorem \ref{thm:NoiStraPubSize}]
It suffices to show $\frac{\tilde p_h - p_h}{\sqrt{\widetilde{V}_h}} \stackrel{d}{\rightarrow} \mathcal{N}(0, 1)$ for all $h$.
By the finite-population CLT in Theorem \ref{thm: SRS_CLT}, we know
\begin{equation*}
    \frac{\hat p_h - p_h}{\sqrt{\operatorname{Var}(\hat p_h)}} \stackrel{d}{\rightarrow} \mathcal{N}(0, 1).
\end{equation*}
Since $\tilde p_h = \hat p_h + e_h $ where $e_h \sim \mathcal{N}(0, \frac{1}{2 \rho n_h^2})$, we have
\begin{equation}
    \frac{\tilde p_h - p_h}{\sqrt{\operatorname{Var}(\tilde p_h)}} \stackrel{d}{\rightarrow} \mathcal{N}(0, 1)
    \label{nas:an}
\end{equation}
where 
\begin{equation*}
    \operatorname{Var}(\tilde p_h) = \operatorname{Var}(\hat p_h) + \frac{1}{2 \rho n_h^2}.
\end{equation*}
Let
\begin{equation}
\begin{aligned}
\widetilde{V}_h
    & =  \left(\frac{N_h-n_h}{N_h}\right)\frac{\tilde p_h(1-\tilde p_h) +\frac{1}{2\rho n_h^2}}{n_h-1} + \frac{1}{2\rho n_h^2}\\
    & = \widehat{\operatorname{Var}}(\hat p_h) + \left(\frac{N_h - n_n}{N_h}\right) \frac{e_h - 2 \tilde p_h e_h - e_h^2 + \frac{1}{2 \rho n_h^2}}{n_h - 1} + \frac{1}{2 \rho n_h^2}.
\end{aligned}
\label{nas:var_tilde}
\end{equation}
Since $e_h \sim \mathcal{N}(0, \frac{1}{2 \rho n_h^2})$, we have $e_h = O_P(\frac{1}{\sqrt{\rho}n_h})$. Then, the second term of (\ref{nas:var_tilde}) is $O_P(\frac{1}{\sqrt{\rho}n_h^2})$,
and thus,
$\widetilde{V}_h-\operatorname{Var}(\tilde p_h) = \widehat{\operatorname{Var}}(\hat p_h) -  \operatorname{Var}(\hat p_h) + O_P\left(\frac{1}{\sqrt{\rho}n_h^2}\right)$.
Note that $\widehat{\operatorname{Var}}(\hat p_h)$ is of order $\frac{1}{n_h}$, and that by Lemma \ref{prop: consistency}, $\widehat{\operatorname{Var}}(\hat p_h) \stackrel{p}{\rightarrow} \operatorname{Var}(\hat p_h)$. 
Therefore, 
$\widetilde{V}_h\stackrel{p}{\rightarrow} \operatorname{Var}(\tilde p_h)$, and thus,
$\widetilde{V} \stackrel{p}{\rightarrow} \operatorname{Var}(\tilde p)$.

Combining the consistency of $\widetilde{V}$ with (\ref{nas:an}), we have 
\begin{equation}
    \frac{\tilde p_h - p_h}{\sqrt{\widetilde{V}_h}} \stackrel{d}{\rightarrow} \mathcal{N}(0, 1)
\end{equation}
by Slutsky's Theorem. Then, $\frac{\tilde p - p}{\sqrt{\widetilde{V}}} \stackrel{d}{\rightarrow} \mathcal{N}(0, 1)$.
Therefore, the confidence interval given by $p \pm z_{1-\alpha/2}\sqrt{\widetilde V}$ has asymptotic coverage level $1 - \alpha$.

\end{proof}

\subsection{Proof of Theorem \ref{thm:NoiPopPubSize}}

\begin{proof}
Since $\frac{\hat p - p}{\sqrt{\operatorname{Var}(\hat p)}} \stackrel{d}{\rightarrow} \mathcal{N}(0, 1)$ and $\tilde p = \hat p + \mathcal{N}(0, \Delta_p^2 / 2\rho_1)$ with $\Delta_p = \max_h\frac{w_h}{n_h}$, it follows that
\begin{equation*}
    \frac{\tilde p - p}{\sqrt{\operatorname{Var}(\tilde p)}} \stackrel{d}{\rightarrow} \mathcal{N}(0, 1),
\end{equation*}
and
\begin{equation*}
     \operatorname{Var}(\tilde p) = \operatorname{Var}(\hat p) + \frac{\Delta_p^2}{2\rho_1}.
\end{equation*}
In Algorithm \ref{alg:poplevel}, we set 
\begin{equation}
    \widetilde{V} = \widehat{\operatorname{Var}}(\hat p) + \frac{\Delta_{ p}^2}{2\rho_1} + e_{V}, 
\end{equation}
where $e_{V} \sim \mathcal{N}(0, \frac{\Delta_V^2}{2\rho_2})$ with $\Delta_{V} = \max_h\left(\frac{C_h}{n_h}\left(1-\frac{1}{n_h}\right) \right)$ and $C_h = w_h^2 \frac{N_h-n_h}{N_h}\frac{1}{n_h-1}$. 
Since $\Delta_V = O(\frac{1}{\max_h n_h^2})$, we have $e_V = O_P(\frac{1}{\max\limits_h \sqrt{\rho_2}n_h^2})$. 
Thus,
$\widetilde{V}-\operatorname{Var}(\tilde p) = \widehat{\operatorname{Var}}(\hat p) -  \operatorname{Var}(\hat p) + O_P\left(\frac{1}{\max_h\sqrt{\rho}n_h^2}\right)$.
Since $\widehat{\operatorname{Var}}(\hat p) \stackrel{p}{\rightarrow} \operatorname{Var}(\hat p)$ by finite-population CLT, we have  $\widetilde{V} \stackrel{p}{\rightarrow} \operatorname{Var}(\tilde p)$.

Therefore, by Slutsky's Theorem, \begin{equation}
    \frac{\tilde p - p}{\sqrt{\widetilde{V}}} \stackrel{d}{\rightarrow} \mathcal{N}(0, 1).
\end{equation}
Then, the confidence interval given by $p \pm z_{1-\alpha/2}\sqrt{\widetilde V}$ has the asymptotic coverage level $1 - \alpha$.

\end{proof}

\subsection{Proof of Theorem \ref{thm:consistency}}
\begin{proof}
For $\tilde n \sim \mathcal{N}(n, \frac{1}{2\rho_2})$, by Proposition \ref{prop:condmom},
we derive the $k$th-order Taylor series of the conditional expectation of $\tilde p$ given $S = \{1\leq \tilde n\leq 2n-1\}$:
\begin{equation}
\begin{aligned}
    \mathbb{E}\left(\tilde p \mid S \right) & =  p \sum_{j = 0}^{k} \frac{(2j-1)!!}{n^{2j}(2\rho_2)^j} + O\left(\frac{1}{n^{2k+1}\rho_2^{k+1}}\right).
\label{tl:mean}
\end{aligned}
\end{equation}
For example, when $k=2$, 
\begin{equation}
\begin{aligned}
    \mathbb{E}\left(\tilde p \mid S \right) & =  p \left( 1 + \frac{1}{2n^2\rho_2} + \frac{3}{4n^4\rho_2^2} \right) +  O\left(\frac{1}{n^{5}\rho_2^{3}}\right).
\end{aligned}
\label{taylor:mean2}
\end{equation}
To obtain a Taylor expansion for the conditional variance, we plug
\begin{equation*}
  \mathbb{E}\left(\frac{1}{\tilde n}\mid S \right)       
         = \frac{1}{n}\sum_{j = 0}^{k} \frac{(2j-1)!!}{n^{2j}(2\rho_2)^{j}}   + O\left(\frac{1}{n^{2k+2}\rho_2^{k+1}}\right)
  \label{eq:1stmomofnt}  
\end{equation*}
 and 
\begin{equation*}
  \mathbb{E}\left(\frac{1}{\tilde n^2}\mid S \right) = \frac{1}{n^2}\sum_{j = 0}^{k} \frac{(2j+1)!!}{n^{2j}(2\rho_2)^{j}}   + O\left(\frac{1}{n^{2k+2}
  \rho_2^{k+1}}\right)
  \label{eq:2ndmomofnt}  
\end{equation*} 
into
\begin{equation*}
\begin{aligned}
    \operatorname{Var}(\tilde p\mid S ) & = 
    \mathbb{E}(\tilde p^2 \mid S )-   (\mathbb{E}(\tilde p \mid S ))^2=
    \mathbb{E}\tilde c^2\mathbb{E}\left(\frac{1}{\tilde n^2}\mid S \right) - (\mathbb{E}(\tilde p \mid S ))^2,
\end{aligned}
\end{equation*}
by which we derive a general expansion for the conditional variance:
\begin{equation}
\begin{aligned}
    \operatorname{Var}(\tilde p\mid S) 
  & = \operatorname{Var}(
    \hat p) \sum_{j = 0}^{k} \frac{(2j+1)!!}{n^{2j}(2\rho_2)^{j}}   +
      p^2 \left( 
    \sum_{j = 0}^{k} \frac{(2j+1)!!}{n^{2j}(2\rho_2)^{j}} - \left(\sum_{j = 0}^{k} \frac{(2j-1)!!}{n^{2j}(2\rho_2)^{j}}   \right)^2 \right)\\
      & \quad + \frac{1}{2\rho_1}\sum_{j = 0}^{k} \frac{(2j+1)!!}{n^{2j+2}(2\rho_2)^{j}}  +  O\left(\frac{1}{n^{2k}
  \rho_2^{k+1}}\right) +O\left(\frac{1}{n^{2k+2}
  \rho_1\rho_2^{k+1}}\right).
\label{eq:varofptilde}
\end{aligned}
\end{equation}
When $k = 2$,
\begin{equation}
\begin{aligned}
    \operatorname{Var}(\tilde p\mid S)
    & = \operatorname{Var}(\hat p) \left(1 +  \frac{3}{2n^2\rho_2} + \frac{15}{4n^4\rho_2^2}\right) + p^2\left(\frac{1}{2n^2\rho_2} + \frac{2}{n^4\rho_2^2}- \frac{6}{8n^6\rho_2^3} - \frac{9}{16n^8\rho_2^4}\right)  \\
    & \quad +
    \frac{1}{2\rho_1} \left(\frac{1}{n^2}+  \frac{3}{2n^4\rho_2}+ \frac{15}{4n^6\rho_2^2}\right)+ O\left(\frac{1}{n^4\rho_2^3}\right) + O\left(\frac{1}{n^6\rho_1\rho_2^3}\right).
\end{aligned}
\label{taylor:var2}
\end{equation}

Based on Taylor expansion with $k = 2$ for both conditional mean and variance given in (\ref{taylor:mean2}) and (\ref{taylor:var2}), 
under the condition $\frac{1}{\rho_1 n} = o(1)$ and $\frac{1}{\rho_2 n} = o(1)$, we have
\begin{equation*}
     \mathbb{E} (\tilde p \mid S) = p +  o\left( \frac{1}{n}\right)
\end{equation*}
and 
\begin{equation*}
     \operatorname{Var} (\tilde p \mid S)= \operatorname{Var}(\hat p) + o\left( \frac{1}{n}\right).
\end{equation*}
Then, $\tilde p \mid S$ is asymptotically unbiased.
Note that $\operatorname{Var}(\hat p)$ is of order  $\frac{1}{n}$ and thus $\tilde p \mid S$ has a vanishing variance. Therefore, $\tilde p \mid S$ converges to $p$ in probability. Note also that $\Pr(S) \rightarrow 1$ as $n \rightarrow \infty$, then for any $\epsilon > 0$, 
\begin{equation*}
    \Pr(|\tilde p - p |> \epsilon) =  \Pr(|\tilde p - p | > \epsilon \mid S) + \Pr(|\tilde p - p | > \epsilon \mid S^\mathsf{c}) \rightarrow 0.
\end{equation*}
That is, $\tilde p$ is a consistent estimator for $p$.

\end{proof}

\subsection{Proof of Theorem \ref{thm:NoiStraPrivSize}}
To prove Theorem \ref{thm:NoiStraPrivSize}, we need the following theorem and lemmas.
\begin{theorem}[Theorem 1, \cite{ratiov2013onthe}]
Let $X$ be a normal random variable with positive mean $\mu_x$, variance $\sigma_x^2$
and coefficient of variation $\delta_x = \sigma_x / \mu_x$ such that $ 0 < \delta_x < \lambda  \leq 1$, where $\lambda$ is
a known constant. For every $\epsilon > 0$, there exists $\gamma(\epsilon) \in (0, \sqrt{\lambda^2 - \delta_x^2})$ and also a
normal random variable $Y$ independent of $X$, with positive mean $\mu_y$, variance $\sigma_y^2$
and coefficient of variation $\delta_y = \sigma_y / \mu_y$ that satisfy the conditions,
\begin{equation}
    0 < \delta_y \leq  \gamma(\epsilon) \leq \sqrt{\lambda^2 - \delta_x^2} < \lambda
\label{lem:cond}
\end{equation}
for which the following result holds. Any $z$ that belongs to the interval
\begin{equation*}
    I = \left[\beta - \frac{\sigma_z}{\lambda}, \beta + \frac{\sigma_z}{\lambda} \right],
\end{equation*}
where $\beta = \mu_x / \mu_y$, $\sigma_z = \beta \sqrt{\delta_x^2 + \delta_y^2}$, satisfies that 
\begin{equation*}
    |G(z) - F_Z(z)| < \epsilon,
\end{equation*}
where $G(z)$ is the cumulative distribution function of $\mathcal{N}(\beta, \sigma^2_z)$, and $F_Z$ is that of $Z = X/Y$. Note that once a given $Y$ fulfills the closeness between the corresponding
$G$ to $F_Z$ , any other random variables with a smaller coefficient of variation will satisfy this result too.
\label{lem:ratio}
\end{theorem}

\begin{lemma}[] For a population of size $N$, let $p$ be the true proportion in the population with the attribute of interest. Consider simple random sampling with sample size $n$. Let
 $Z^* \sim \mathcal{N}(p, V)$ where $V  = \left( \frac{N-n}{N-1} \right) \frac{p(1-p)}{n} + \frac{1}{2\rho_1 n^2}  + \frac{p^2}{2\rho_2 n^2} $.
If $\rho_1 = \omega(1/n^2)$ and $\rho_2 = \omega(1/n)$, as $N - n$ and $n$ both tend to infinity, then for any $z \in (0, 2p)$, 
\begin{equation}
    |F_{\tilde p}(z) - F_{Z^*}(z)| \rightarrow 0.
\end{equation}

\label{lem:CDF}
\end{lemma}
\begin{proof}
By the CLT in Theorem \ref{thm: SRS_CLT}, we know that $\hat p \sim \mathcal{AN}(p, \operatorname{Var}(\hat p))$. Recall that $\tilde c = n\hat p + \mathcal{N}(n, \frac{1}{2\rho_1})$, then $\tilde c \sim \mathcal{AN}(np, n^2\operatorname{Var}(\hat p) + \frac{1}{2\rho_1})$.

Let $\tilde X \sim \mathcal{AN}(np, n^2\operatorname{Var}(\hat p) + \frac{1}{2\rho_1})$ and $X \sim \mathcal{N}(np, n^2\operatorname{Var}(\hat p) + \frac{1}{2\rho_1})$. Therefore, for any $\epsilon > 0$, there exists some $n_0 = n_0(\epsilon)$ such that for any $x$ and $n > n_0$,
\begin{equation}
    |F_{\tilde X}(x) - F_{X}(x)| < \epsilon,
    \label{eq:1}
\end{equation}
where $F$ denotes the cumulative density function.
Let $Y \sim \mathcal{N}(n, \frac{1}{2\rho_2})$, $\tilde Z = \tilde X/Y$ and $Z = X/Y$, then 
\begin{equation*}
    F_{\tilde Z}(z) = \Pr \left(\frac{\tilde X}{Y} < z\right) 
    = \Pr (\tilde X < Yz)
    = \int_{-\infty}^{\infty} F_{\tilde X}(yz)f_y(y)dy,
\end{equation*}
where $f_y(y)$ is the density function of $Y$.
From (\ref{eq:1}), $ |F_{\tilde X}(yx) - F_{X}(yx)| < \epsilon$. 
It follows that, 
\begin{equation*}
    \int_{-\infty}^{\infty} (F_{X}(yx) - \epsilon) f_y(y)dy < \int_{-\infty}^{\infty} F_{\tilde X}(yx)f_y(y)dy < \int_{-\infty}^{\infty} (F_{X}(yx) + \epsilon)f_y(y)dy,
\end{equation*}
which is equivalent to
\begin{equation*}
    \left|\int_{-\infty}^{\infty}F_{\tilde X}(yx)f_y(y)dy - \int_{-\infty}^{\infty}F_{X}(yx)f_y(y)dy\right| < \epsilon,
\end{equation*}
i.e., 
\begin{equation}
    |F_{\tilde Z}(z) - F_{Z}(z) | < \epsilon.
\end{equation}

Let $\delta_x$ and $\delta_y$ be the coefficients of variation of $X$ and $Y$, respectively, then $\delta_x^2 = (\operatorname{Var}(\hat p) + \frac{1}{2 \rho_1 n^2}) / p^2$ and $\delta_y^2 = \frac{1}{2 \rho_2 n^2}$. Under the condition $\frac{1}{\rho_1n} = o(1)$, we have $\delta_x^2 = O(\frac{1}{n})$ since $\operatorname{Var}(\hat p) = O(\frac{1}{n})$. Under the condition $\frac{1}{\rho_2 n} = o(1)$, we know $\delta^2_y = o(\frac{1}{n})$ and then $\delta_y = o(\delta_x)$.
When $n$ is sufficiently large, $\delta_y$ is sufficiently small. Let $\lambda = \sqrt{\delta_x^2+ 2\delta_y^2}$ and $F_{Z^*}(z)$ be the distribution function of $Z^* \sim \mathcal{N}(p, \operatorname{Var}(\hat p) + \frac{1}{2\rho_1 n^2}  + \frac{p^2}{2\rho_2 n^2} )$. By Lemma \ref{lem:ratio}, for a
normal random variable $Y$ independent of $X$, with small enough $\delta_y$, the condition (\ref{lem:cond}) is satisfied and we have
\begin{equation}
    |F_{Z}(z)  - F_{Z^*}(z) | < \epsilon,
\end{equation}
for any $z \in I = \left[ p - \frac{\sigma_{z^*}}{\lambda}, p + \frac{\sigma_{z^*}}{\lambda} \right]$ where  $\sigma_{z^*} = p \sqrt{\delta_x^2 + \delta_y^2}$. Hence, for $z \in I$, 
\begin{equation}
    |F_{\tilde Z}(z)  - F_{Z^*}(z) | < |F_{\tilde Z}(z)  - F_{Z}(z) |  + |F_{Z}(z)  - F_{Z^*}(z) | < 2\epsilon.
\end{equation}
Note also that as $n \rightarrow \infty$, $\frac{\sigma_{z^*}}{\lambda} \rightarrow p$, and the limit of $I$ is $(0, 2p)$. 

So far, we have shown that as $n$ goes to infinity, under the conditions $\frac{1}{\rho_1 n} = o(1)$ and $\frac{1}{\rho_2 n} = o(1)$, 
for $z \in I_h$,
\begin{equation}
    |F_{\tilde Z}(z)  - F_{Z^*}(z) | \rightarrow 0.
\end{equation}
\end{proof}

\begin{lemma}
\label{lem:conv}
   Let $Z_1, ..., Z_H$ and $Z_1^*, ..., Z_H^*$ be independent continuous random variables which depend on $n$. Let $F$ denote the distribution function. As $n \rightarrow \infty$, if
\begin{equation*}
    |F_{Z_h}(z) - F_{Z_h^*}(z)| \rightarrow 0
\end{equation*}
holds for any $h = 1, ..., H$ and $z$ in an interval $(a_h, b_h)$
and $\Pr(Z_h \in (a_h, b_h)) \rightarrow 1$, $ \Pr(Z^*_h \in (a_h, b_h)) \rightarrow 1
$.
 Then,
\begin{equation*}
\left|F_{\sum_{h = 1}^H c_h Z_h}(z) - F_{\sum_{h = 1}^H c_h Z_h^*}(z) \right| \rightarrow 0 
\end{equation*}
for any $z \in \left(\sum_{h = 1}^H c_h a_h, \sum_{h = 1}^H c_h b_h \right)$, where $c_h$'s are constants.
\end{lemma}

\begin{proof}
It suffices to show that, for any $z \in (a_1c_1 + a_2c_2, b_1c_1 + b_2c_2)$,
\begin{equation*}
    \left|F_{c_1Z_1 + c_2Z_2}(z)- F_{c_1Z^*_1 + c_2Z^*_2}(z) \right| \rightarrow 0
\end{equation*}
as $n \rightarrow \infty$.
We have
\begin{equation}
\begin{aligned}
    & \quad F_{c_1Z_1 + c_2Z_2}(z) \\
    & = \Pr (c_1Z_1 + c_2Z_2 < z) \\
    & = \Pr\left(Z_1 < \frac{z - c_2Z_2}{c_1}\right) \\
    & = \int_{\mathbb{R}} F_{Z_1}\left(\frac{z - c_2x}{c_1}\right)f_{Z_2}(x)dx \\
    & = \int_{\mathbb{R}} \left[F_{Z_1}\left(\frac{z - c_2x}{c_1}\right) - F_{Z^*_1}\left(\frac{z - c_2x}{c_1}\right) \right] f_{Z_2}(x)dx  + \int_{\mathbb{R}} F_{Z^*_1}\left(\frac{z - c_2x}{c_1}\right)f_{Z_2}(x)dx \\
    & = \int_{\mathbb{R}} \left[F_{Z_1}\left(\frac{z - c_2x}{c_1}\right) - F_{Z^*_1}\left(\frac{z - c_2x}{c_1}\right) \right] f_{Z_2}(x)dx  + F_{c_1Z^*_1 + c_2Z_2}(z).
\end{aligned} 
\label{eq:conv}
\end{equation}
When $a_1 < (z - c_2x)/c_1 < b_1$, we know
\begin{equation}
    \left|F_{Z_1}\left(\frac{z - c_2x}{c_1}\right) - F_{Z_1^*}\left(\frac{z - c_2x}{c_1}\right)\right| \rightarrow 0.
    \label{eq:conv2}
\end{equation}

Since $F_{Z_1}(b_1) - F_{Z_1}(a_1) \rightarrow 1$ and $F_{Z^*_1}(b_1) - F_{Z^*_1}(a_1) \rightarrow 1$, for any $a < a_1$, it holds that $F_{Z_1}(a) \rightarrow 0$ and $F_{Z^*_1}(a) \rightarrow 0$, and for any $b > b_1$, $F_{Z_1}(b) \rightarrow 1$ and $F_{Z^*_1}(b) \rightarrow 1$. 
Thus, (\ref{eq:conv2}) also holds when $(z - c_2x)/c_1 $ is outside $(a_1, b_1)$. 
Therefore, the first term of the right-hand side of (\ref{eq:conv}) converges to 0. Then

\begin{equation*}
   \left|F_{c_1Z_1 + c_2Z_2}(z)- F_{c_1Z^*_1 + c_2Z_2}(z) \right| \rightarrow 0.
\end{equation*}
Similarly, we have
\begin{equation*}
   \left|F_{c_1Z^*_1 + c_2Z_2}(z)- F_{c_1Z^*_1 + c_2Z^*_2}(z) \right| \rightarrow 0.
\end{equation*}
By the triangle inequality,
\begin{equation*}
   \left|F_{c_1Z_1 + c_2Z_2}(z)- F_{c_1Z^*_1 + c_2Z^*_2}(z) \right| \rightarrow 0.
\end{equation*}

\end{proof}

\begin{proof}[Proof of Theorem \ref{thm:NoiStraPrivSize}]
By Lemma \ref{lem:CDF}, 
for each stratum, under the conditions $\rho_1 = \omega(1/n_h)$ and $\rho_2 = \omega(1/n_h)$, the distribution function of $\tilde p_h$ converges to that of $\mathcal{N}(p_h, V_h)$ in the interval $(0, 2 p_h)$ where 
\begin{equation}
\label{eq:true varh priv}
    V_h = \left( \frac{N_h-n_h }{N_h -1} \right) \frac{p_h (1-p_h )}{n_h } + \frac{1}{2\rho_1 n_h^2}  + \frac{p_h ^2}{2\rho_2 n_h^2}.
\end{equation}
Let $p^* \sim \mathcal{N}(p, V)$ where $V = \sum_{h = 1}^H w_h^2 V_h$.
By Lemma \ref{lem:conv}, in the interval $(0, 2 p)$, we have
\begin{equation}
    |F_{\tilde p}(z)  - F_{p^*}(z) | \rightarrow 0.
    \label{eq:diff}
\end{equation}
where $F_{\tilde p}$ denotes the distribution function of $\tilde p$ designed in Algorithm \ref{alg:{StrNz-PrivSz}} and $F_{p^*}$ is the distribution function of $p^*$.

Let $ L =  p - z_{1-\alpha/2} \sqrt{ V}$, $ U =  p+ z_{1-\alpha/2} \sqrt{V}$, $\tilde L =  p - z_{1-\alpha/2} \sqrt{\widetilde V}$ and $\tilde U =  p+ z_{1-\alpha/2} \sqrt{\widetilde V}$. Note that $L$ and $U$ are constants whereas $\tilde L$ and $\tilde U$ are random variables. 
Provided that $n_h$'s are sufficiently large, $U$ and $L$ lie in the interval where the following hold
 due to (\ref{eq:diff}), 
\begin{equation}
    |F_{\tilde p}( U) - F_{p^*}( U) |   \rightarrow 0
    \label{eq:conv_FU}
\end{equation}
and 
\begin{equation}
    |F_{\tilde p}( L) - F_{p^*}( L) | \rightarrow 0.
\end{equation}

On the other hand, by Theorems \ref{prop:condmom} and \ref{thm:consistency}, we know that $\tilde p_h \stackrel{p}{\rightarrow} p_h$ and $\frac{1}{\tilde n_h} \stackrel{p}{\rightarrow} \frac{1}{n_h}$ under the conditions $\rho_1 = \omega(1/n_h)$ and $\rho_2 = \omega(1/n_h)$. 
By the continuous mapping theorem, $\widetilde V_h\stackrel{p}{\rightarrow} V_h$ as $n_h \rightarrow \infty$, and, hence, $\widetilde V\stackrel{p}{\rightarrow} V$. Therefore, $\tilde U\stackrel{p}{\rightarrow} U$ and $\tilde L\stackrel{p}{\rightarrow} L$. Since $F_{\tilde p}$ is continuous, we have
\begin{equation}
    |F_{\tilde p}(\tilde U) - F_{\tilde p}( U)| \stackrel{p}{\rightarrow} 0
\end{equation}
and 
\begin{equation}
    |F_{\tilde p}(\tilde L) - F_{\tilde p}( L) | \stackrel{p}{\rightarrow} 0.
    \label{eq:conv_FL_tilde}
\end{equation}

Therefore,
\begin{equation*}
    \begin{aligned}
        &\quad \Pr \left( p \in \left(\tilde p - z_{1-\alpha/2}\sqrt{\widetilde V}, \tilde p + z_{1-\alpha/2}\sqrt{\widetilde V}\right)\right) \\
        & = \Pr\left( p - z_{1-\alpha/2}\sqrt{\widetilde V} < \tilde p <p + z_{1-\alpha/2}\sqrt{\widetilde V} \right) \\
        & = \left(F_{\tilde p}(\tilde U) - F_{\tilde p}( U) \right)+ \left( F_{\tilde p}( U)  - F_{p^*}( U) \right) \\ 
        & \quad 
        -\left(F_{\tilde p}(\tilde L) - F_{\tilde p}( L) \right) - \left( F_{\tilde p}( L)  - F_{p^*}( L) \right) +  \left(F_{p^*}( U) - F_{p^*}( L)\right) .
    \end{aligned}
\end{equation*}
Putting together (\ref{eq:conv_FU}) through (\ref{eq:conv_FL_tilde}) and $F_{p^*}( U) - F_{p^*}( L) = 1- \alpha$, we have
\begin{equation*}
    \lim_{n\rightarrow \infty}\Pr \left( p \in \left(\tilde p - z_{1-\alpha/2}\sqrt{\widetilde V}, \tilde p + z_{1-\alpha/2}\sqrt{\widetilde V}\right)\right)
      \rightarrow 1- \alpha.
\end{equation*}

Since $\frac{1}{\tilde n_h} \stackrel{p}{\rightarrow} \frac{1}{n_h}$, under the conditions $\rho_1 = \omega(1/n_h)$ and $\rho_2 = \omega(1/n_h)$, it holds that $\frac{1}{2\rho_1 \tilde n_h^2} = o_P(\frac{1}{n_h})$ and $\frac{\tilde p^2_h}{2\rho_2 \tilde n_h^2} = o_P(\frac{1}{n_h})$. Therefore, the additional error in estimating the conditional variance of $\tilde p$ caused by the injected noise is $O_p\left(\frac{1}{\rho_1n_h^2} + \frac{1}{\rho_2n_h^2}\right)= o_p\left(\frac{1}{n_h}\right)$.

\end{proof}

\begin{acks}[Acknowledgments]
We are grateful for helpful conversations with and comments
from (in no particular order) Rolando Rodriguez, 
Brian Finley, Jörg Drechsler, Gary Benedetto, Michael Freiman, and Justin Doty. 
This project was funded by
the U.S. Census Bureau cooperative agreements CB20ADR0160001.
\end{acks}


\bibliographystyle{imsart-number} 
\bibliography{bibliography}       


\end{document}